\pdfoutput=1

\newif\ifdraft\draftfalse   
\newif\ifanon\anontrue      
\newif\ifcamera\camerafalse 
\newif\iflongrefs\longrefsfalse 
\newif\ifsooner\soonerfalse 
\newif\iflater\laterfalse   
\newif\iffull\fullfalse   
\newif\ifneedspace\needspacefalse 
\newif\ifhighlightnewtext\highlightnewtextfalse 
\makeatletter \@input{texdirectives.tex} \makeatother

\ifcamera
  \documentclass[acmsmall, screen]{acmart}
\else
  \ifanon
  \documentclass[acmsmall,review,anonymous,screen,nonacm]{acmart}
  \else
  \documentclass[acmsmall,review=false,screen=true,nonacm]{acmart}
  \fi
\fi

\AtBeginDocument{%
  }

\ifanon
\setcopyright{none}             
\else
\setcopyright{cc}
\setcctype{by}
\acmDOI{10.1145/3747522}
\acmYear{2025}
\acmJournal{PACMPL}
\acmVolume{9}
\acmNumber{ICFP}
\acmArticle{253}
\acmMonth{8}
\received{2025-02-27}
\received[accepted]{2025-06-27}
\fi

\usepackage{booktabs}   
\usepackage{subcaption} 

\usepackage[normalem]{ulem}

\definecolor{dkblue}{rgb}{0,0.1,0.5}
\definecolor{dkgreen}{rgb}{0,0.4,0}
\definecolor{dkred}{rgb}{0.6,0,0}
\definecolor{dkpurple}{rgb}{0.7,0,1.0}
\definecolor{purple}{rgb}{0.9,0,1.0}
\definecolor{olive}{rgb}{0.4, 0.4, 0.0}
\definecolor{teal}{rgb}{0.0,0.4,0.4}
\definecolor{azure}{rgb}{0.0, 0.5, 1.0}
\definecolor{gray}{rgb}{0.5, 0.5, 0.5}
\definecolor{dkgray}{rgb}{0.3, 0.3, 0.3}
\definecolor{orange}{rgb}{1.0, 0.5, 0.0}
\definecolor{lhorange}{rgb}{1.0, 0.8, 0.6}

\newcommand{\comm}[3]{\ifdraft{{\color{#1}[#2: #3]}}\fi}
\newcommand{\tbd}[1]{\comm{orange}{TBD}{#1}}
\newcommand{\nth}[1]{\comm{orange}{Stretch goal}{#1}}
\newcommand{\ch}[1]{\comm{teal}{CH}{#1}}
\newcommand{\ca}[1]{\comm{olive}{CA}{#1}}
\newcommand{\tw}[1]{\comm{purple}{TW}{#1}}
\newcommand{\da}[1]{\comm{brown}{DA}{#1}}

\newcommand{\meta}[1]{\ifdraft{{\color{dkgray}[#1]}}\fi} 

\usepackage{xspace}

\newcommand*{\EG}{e.g.,\xspace}
\newcommand*{\IE}{i.e.,\xspace}

\newcommand\fstar{F$^{\star}$\xspace}

\newcommand\secrefstar{SecRef$^{\star}$\xspace}

\newcommand\titfstar{\texorpdfstring{F$^{\star}$\xspace}{F*}\xspace}
\newcommand\titsecrefstar{\texorpdfstring{SecRef$^{\star}$\xspace}{SecRef*}}

\newcommand{\cmparrow}{\hspace{-0.35em}\downarrow}
\newcommand{\cmp}[1]{#1\cmparrow}
\newcommand{\bakarrow}{\hspace{-0.35em}\uparrow}
\newcommand{\bak}[1]{#1\bakarrow}

\newcommand{\inv}{\mathcal{I}}
\newcommand{\prref}{\varphi}
\newcommand{\hrel}{\preccurlyeq}
\newcommand{\threep}{\inv\ \prref\ (\hrel)}
\newcommand{\cinv}{\inv_{c}}
\newcommand{\cprref}{\prref_{c}}
\newcommand{\chrel}{\hrel_{c}}
\newcommand{\cthreep}{\cinv\ \cprref\ (\chrel)}
\newcommand{\every}{\text{\sffamily every$_{\text{\sffamily mref}}$}}

\usepackage{listings}
\usepackage[framemethod=tikz]{mdframed}
\newcommand\maybecolor[1]{\color{#1}}

\ifneedspace
\fi
\definecolor{light-gray}{gray}{0.95}
\newcommand{\ls}[1]{\lstinline{#1}}

\newcommand{\lss}[1]{\lstinline[basicstyle=\footnotesize]{#1}}

\usepackage{combelow} 
\usepackage{tikz}
\usepackage{wrapfig}
\usepackage[inline]{enumitem}
\newlist{inlist}{enumerate*}{1}
\setlist[inlist]{label=(\arabic*)}

\usetikzlibrary{shapes,arrows}
\tikzset{%
  block/.style    = {draw, thick, rectangle, minimum height = 3em,
    minimum width = 3em},
}

\def\Snospace~{\S{}}

\usepackage[skip=0pt, belowskip=-10pt]{caption}

\newcommand{\newtext}[1]{\ifhighlightnewtext{\color{dkgreen}#1}\else#1\fi}
\newcommand{\newremove}[1]{\ifhighlightnewtext{\color{dkred}\sout{#1}}\fi}

\begin{document}


\title{\titsecrefstar: Securely Sharing Mutable References between Verified and Unverified Code in \titfstar}

\ifanon
\author{}
\else
\author{Cezar-Constantin Andrici}
  \affiliation{\institution{MPI-SP}\city{Bochum}\country{Germany}}
  \email{cezar.andrici@mpi-sp.org}
  \orcid{0009-0002-7525-2440}
\author{Danel Ahman}
  \affiliation{\institution{University of Tartu}\ifcamera\city{Tartu}\fi\country{Estonia}}
  \email{danel.ahman@ut.ee}
  \orcid{0000-0001-6595-2756}
\author{C\u{a}t\u{a}lin Hri\cb{t}cu}
  \affiliation{\institution{MPI-SP}\city{Bochum}\country{Germany}}
  \email{catalin.hritcu@mpi-sp.org}
  \orcid{0000-0001-8919-8081}
\author{Ruxandra Icleanu}
  \affiliation{\institution{University of Edinburgh}\ifcamera\city{Edinburgh}\fi\country{UK}}
  \email{r.icleanu@gmail.com}
  \orcid{0009-0004-5510-7080}
\author{Guido Mart\'{i}nez}
  \affiliation{\institution{Microsoft Research}\city{Redmond}\state{WA}\country{USA}}
  \email{guimartinez@microsoft.com}
  \orcid{0009-0005-5831-9991}
\author{Exequiel Rivas}
  \affiliation{\institution{Tallinn University of Technology}\ifcamera\city{Tallinn}\fi\country{Estonia}}
  \email{exequiel.rivas@ttu.ee}
  \orcid{0000-0002-2114-624X}
\author{Théo Winterhalter}
  \affiliation{\institution{Inria Saclay}\ifcamera\city{Saclay}\fi\country{France}}
  \email{theo.winterhalter@inria.fr}
  \orcid{0000-0002-9881-3696}
\fi

\ifanon\else
\renewcommand{\shortauthors}{Andrici et al.}
\fi

\begin{abstract}
We introduce \secrefstar, a secure compilation framework protecting stateful
programs verified in \fstar against linked unverified code, with which
the program dynamically shares ML-style mutable references.
To ease program verification in this setting, we \newremove{propose a way of tracking}\newtext{track}
which references are shareable with the unverified code, and
which ones are not shareable and whose contents are thus guaranteed to be
unchanged after calling into unverified code.
This universal property of non-shareable references is exposed in the interface
on which the verified program can rely when calling into unverified code.
The remaining refinement types and pre- and post-conditions that the verified
code expects from the unverified code are converted into dynamic checks about
the shared references by using higher-order contracts.
We prove formally in \fstar that this strategy ensures sound and secure
interoperability with unverified code.
Since \secrefstar is built on top of the Monotonic State effect of
\fstar, these proofs rely on the first monadic representation for this effect,
which is a contribution of our work that can be of independent interest.
Finally, we use \secrefstar to build a simple cooperative multi-threading
scheduler that is verified and that securely interacts with unverified threads.
\end{abstract}

\ifanon\else

\ccsdesc[500]{Software and its engineering~Software verification}
\ccsdesc[500]{Software and its engineering~Compilers}

\keywords{secure compilation, formal verification, proof assistants, monotonic state}
\fi



\maketitle

\section{Introduction}
\label{sec:intro}
Verifying stateful programs using proof-oriented languages like Rocq, Lean, Isabelle/HOL, Dafny, and \fstar{}
is being applied to increasingly realistic projects~\cite{compcert,certikos,record,everparse,DBLP:conf/cpp/ArasuRRSFHPR23,
linear-dafny,veribetrkv,BhargavanB0HKSW21, HoPBB22,
ZakowskiBYZZZ21, Appel16, MurrayMBGBSLGK13, KleinAEHCDEEKNSTW10}.
In such projects,
one notices that verified code calls into, or is called by, unverified code.
These calls are necessary to integrate verified code with existing systems,
but usually they have an assumed specification,
making them a potential source of unsoundness and security vulnerabilities,
since the unverified code can inadvertently or maliciously break the internal invariants of the verified code.

While this is a general problem in any proof-oriented language,
in this work we work with \fstar{} and represent stateful programs using a monadic effect
called Monotonic State~\cite{preorders}, instantiated to obtain first-order
ML-style mutable references, which are dynamically allocated and unforgeable,
which cannot be deallocated, and whose type cannot be changed.
%
Monotonic State has facilitated the verification of large projects that were later integrated into Firefox,
Linux, and Windows~\cite{haclstar,evercrypt,everparse3d}, yet it offers
no guarantees that linking a verified program with unverified code is sound or secure.

The most challenging scenario for ensuring soundness involves dynamic sharing of
dynamically allocated mutable references with unverified code.
%
This is so because sharing a reference is a transitive property:
the contents of the reference also gets shared, including any references reachable from it.
Sharing is also permanent: once shared, a reference is shared forever as
it can be stashed away by unverified code,
and writes to a shared reference make the newly written values shared as well.
It is thus challenging to track which references are shared and which ones are not,
so that one could still apply static verification.
Let's look at a simple example that illustrates these points:

\begin{wrapfigure}[8]{l}{.48\textwidth}
  \vspace{-0.50\intextsep}
\begin{lstlisting}[numbers=left,escapechar=\$,xleftmargin=13pt,framexleftmargin=10pt,numberstyle=\color{dkgray},xrightmargin=0pt]
let prog (lib : (ref (ref int) -> cb:(unit -> unit))) =
  let secret : ref int = alloc 42 in
  let r : ref (ref int) = alloc (alloc 0) in  $\label{line:prog:alloc}$
  let cb = lib r in $\label{line:prog:fst_lib_call}$
  r := alloc 1; $\label{line:prog:change_r}$
  cb (); $\label{line:prog:snd_lib_call}$
  assert (!secret == 42) $\label{line:prog:assert}$
\end{lstlisting}
\end{wrapfigure}

\noindent
It consists of a program that takes as argument an unverified library.
The library gets as input a reference to a reference to an integer,
and returns a callback \ls{cb} to the program.
The program allocates a reference \ls{r} (line \autoref{line:prog:alloc}) and
passes it to the library (line \autoref{line:prog:fst_lib_call}).
The library can store \ls{r}
and the reference \ls{r} points to, and read and modify them in later callbacks.
Thus, after the first call into the library (line \ref{line:prog:fst_lib_call}),
the program and the library share both \ls{r} and the reference \ls{r} points to.
Between the initial call into the library and the call into the callback \ls{cb}
(line \ref{line:prog:snd_lib_call}), the program modifies \ls{r} pointing it to a newly allocated reference
(line \ref{line:prog:change_r}), and thus this new reference also becomes
shared, even if not passed to \ls{cb} directly---the
sharing happens by updating an already shared reference.
So in this example, a total of three references allocated by the program are shared with the library,
even if only one reference was explicitly passed.
Here is an adversarial library that writes to all shared references
every time the callback is invoked:
\begin{lstlisting}
let lib r = let g : ref (list (ref int)) = alloc [] in
  (fun () -> g := !r :: !g;          (** saves where $r$ points to each time the callback is called **)
         iter (fun r' -> r' := 0) !g; (** writes to all these shared integer references **)
         r := alloc 0)           (** points $r$ to a new location **)
\end{lstlisting}

To be able to verify a stateful program that calls into unverified code, for
each such call one has to specify which of the program's references get modified,
also called the footprint of a call.
If one cannot determine such a footprint, then in order to ensure soundness, one is
forced to conservatively assume that the call can modify {\em all} references,
which makes verification very challenging, because one loses almost all
properties of references. For instance, in \ls{prog}, we would lose
that \ls{secret} equals $42$, and we would not be able to prove the assertion on
line \ref{line:prog:assert}. 
Yet, as explained above,
because references are dynamically allocated, and because of the transitive and
persistent nature of their dynamic sharing, it is hard to determine which
references are shared and which ones not.

\newremove{
To our knowledge, no existing technique allows for the sound verification of
stateful programs that can be linked with arbitrary unverified code 
while sharing dynamically allocated mutable references with it.}
\newremove{
A simpler case considered by \mbox{\citet{DreyerNB12}} (also imaginable for other
separation logic work~\mbox{\cite{Reynolds02}})
%
is to have no passing of references to
unverified code (\IE no dynamic sharing), but this cannot always be ensured,
because existing code cannot always be clearly separated to avoid passing references.
While one could consider rewriting the code to avoid reference passing by 
wrapping them in closures~\mbox{\cite{SwaseyGD17, SammlerGDL20}},
%
or in some cases 
to only pass references at abstract
types,\ca{for type abstraction, I have no idea what to cite. I added it because
  of the discussion on how one usually hides mutable state behind an interface in OCaml.}
this would involve refactoring the unverified code to use a different interface.
Other techniques allow passing references by checking and restoring the
properties of the not passed references after a call into unverified
code~\mbox{\cite{cerise, AgtenJP15}}, but this is done dynamically and adds an extra cost---for \ls{prog} above,
the contents of \ls{secret} would be checked/restored twice, once for each call into the unverified code.}

The solution we propose in this paper allows the sound verification of stateful \newtext{\fstar{}} programs
even if the verified program is linked with unverified code and references are dynamically shared.
Our solution builds on the idea that being shareable is a monotonic property of
references, meaning that it can be treated as a logical invariant once it is established.
Instead of tracking precisely which references cross the verified-unverified boundary, we
overapproximate by adding a computationally irrelevant labeling mechanism 
in which fresh references are labeled as \ls{private}, and that
allows the user to dynamically label a reference as \ls{shareable},
which is permanent.
Our overapproximation is sound since we statically verify that only references that
are labeled as \ls{shareable} can cross into unverified code.
In our example, this requires the user to add the following calls to \ls{label_shareable}:
\begin{lstlisting}[numbers=left,escapechar=\$,xleftmargin=13pt,framexleftmargin=10pt,numberstyle=\color{dkgray}]
let safe_prog lib =
  let secret = alloc 42 in
  let r = alloc (alloc 0) in $\setlength{\fboxsep}{1pt}\colorbox{lhorange}{\sffamily label\_shareable (!r); label\_shareable r;}$
  let cb = lib r in 
  let v = alloc 1 in $\setlength{\fboxsep}{1pt}\colorbox{lhorange}{\sffamily label\_shareable v;}$ r := v; $\label{line:safe_prog:alloc_share}$
  cb ();
  assert (!secret == 42) $\label{line:safe_prog:assert}$
\end{lstlisting}
\newtext{To correctly overapproximate,}
we also maintain a global invariant
on the heap stating that shareable references
always point only to other shareable references---preventing one from accidentally
pointing a shareable reference to a not shareable reference (\EG the original \ls{r:=alloc 1}
would not be allowed on line \ref{line:safe_prog:alloc_share} because the fresh reference returned by
\ls{alloc 1} is labeled \ls{private}).
The labels simplify specifying what the unverified code can do---it can
modify only \ls{shareable} references---which
allows preserving the properties of private references
(\EG one is able to prove the assertion on line \ref{line:safe_prog:assert} since \ls{secret} is labeled \ls{private}).
%
%
%
\newtext{
Note that we have only one notion of references for which we keep track of
how the labels change dynamically.
}

We have built a secure compilation framework, \secrefstar{}, around the main
ideas above and inspired by previous work for the IO effect~\cite{sciostar}.
\newremove{Similarly to this work,}\newtext{Like them,} we shallowly embed the source and target language
in \fstar{}, in our case on top of Labeled References.
The only \newremove{side }effect our languages have is state, \newremove{and
they both have}\newtext{with} first-order ML-style mutable references.
\newtext{The difference between the languages is the interface through which
the program communicates with the unverified code:
a source interface has
concrete 
refinement types and pre- and post-conditions to model verified code, while a target interface
has them abstract to model unverified code.}
Part of \secrefstar is dedicated to statically verifying stateful source programs
like the one above against rich interfaces assumed about linked unverified target code that
include refinement types as well as pre- and post-conditions.
These interfaces expose a universal property of unverified code,
which cannot modify \ls{private} references.
However, to establish any logical predicates on \ls{shareable} references, which can
cross the verified-unverified boundary, \secrefstar{} \newremove{has to add}\newtext{adds} dynamic checks,
\newremove{which it does}using higher-order contracts.

This paper makes the following \textbf{contributions:}

\begin{itemize}[leftmargin=13pt,nosep,label=$\blacktriangleright$]
  \item We introduce Labeled References, a computationally irrelevant labeling mechanism
  built on top of Monotonic State~\cite{preorders}
  that allows one to verify stateful \newtext{\fstar{}} programs that share ML-style mutable references with unverified code.
  The labeling mechanism tracks which references are shareable with the unverified code and which ones are not.
  The labels simplify specifying which references can be modified by a call
  into unverified code---\IE only shareable references can be modified.
  Our solution allows one to pass between the verified and the unverified code
  mutable data structures such as linked lists
  and callbacks that encapsulate and modify references.
  %

  \item We build a secure compiler that takes a program verified using Labeled References and protects
  it so that it can be securely linked with unverified code, which can only modify shareable references.
  The compiler converts some of the specifications
  at the boundary between verified and unverified code into dynamic checks by adding higher-order contracts,
  and we verified in \fstar{} the soundness of this enforcement with respect to the specification of the program.


\item We also provide a machine-checked proof in \fstar{} that a stateful program
  verified using Labeled References, then compiled with \secrefstar{},
  can be securely linked with arbitrary unverified code.
  For this we prove Robust Relational Hyperproperty Preservation (RrHP), which
  is the strongest secure compilation criterion of \citet{AbateBGHPT19}, and in
  particular stronger than full abstraction.


\item We give the first monadic representation of Monotonic State by
  using a free monad, which overcomes a limitation in the original work of~\citet{preorders}.
  We also provide a machine-checked soundness proof in \fstar for Monotonic
  State differently than the original paper proof of \citet{preorders}.
  While we expect the monadic representation of Monotonic State to be generally useful,
  even beyond this paper, it is also crucial for our mechanized proofs above.


%


  \item Finally, we put our framework into practice by developing a case
  study inspired by a simple cooperative multi-threading model. In
  this case study, we showcase that we can statically verify a simple scheduler
  even when it interacts with and manages unverified threads.
\end{itemize}

\paragraph{Outline.}
\autoref{sec:keyideas} starts with the key ideas of \secrefstar{}.
\autoref{sec:monotonic-state} presents our monadic representation and soundness
proof for Monotonic State.
\autoref{sec:labeled-references} builds Labeled References on top of Monotonic State.
\autoref{sec:hoc} then presents our higher-order contracts.
\autoref{sec:secure-compilation} shows how all pieces of \secrefstar{} come
together and how we prove that \secrefstar{} satisfies soundness and RrHP.
\autoref{sec:case-study} presents the cooperative multi-threading case study.
%
Finally, \autoref{sec:related} discusses related work and \autoref{sec:future}
conclusions and future work.
%
We believe \secrefstar{} is a significant step towards building a
secure compiler from \fstar{} to OCaml.

%
%

\section{Key Ideas of \secrefstar{}}
\label{sec:keyideas}

In this section, \autoref{sec:key-verification} introduces the
running example of an autograder sharing references with unverified code, and
showing how to verify it.
\autoref{sec:key-ucode} then explains how we represent unverified code and
\autoref{sec:key-hoc} how we add higher-order contracts.
%
Lastly, beyond \ls{private} and \ls{shareable},
\autoref{sec:key-encapsulated} introduces a third label,
\ls{encapsulated}, for references that are not shareable with the unverified code,
but that can be indirectly updated by verified callbacks
that are passed to the unverified code.

\begin{figure}[h]
  \begin{lstlisting}[numbers=left,escapechar=\$,xleftmargin=13pt,framexleftmargin=10pt,numberstyle=\color{dkgray}]
type linked_list (a:Type) = | LLNil -> linked_list a | LLCons : a -> ref (linked_list a) -> linked_list a
type student_hw = ll:ref (linked_list int) -> LR (either unit err) $\label{line:hwtype:start}$
    (requires (fun h_0 -> witnessed (is_shareable ll))) $\label{line:hwpre}$
    (ensures (fun h_0 r h_1 -> modif_only_shareable h_0 h_1 /\ same_labels h_0 h_1 /\ $\label{line:hw-post-modifies}$
                        (Inr? r \/ (sorted ll h_1 /\ no_cycles ll h_1 /\ same_values ll h_0 h_1)))) $\label{line:hwpost1}$ $\label{line:hwtype:end}$
let autograder (test:list int) (hw:student_hw) (grade: mref (option int) grade_preorder) $\label{line:agtype}\label{line:mgrade}$
  : LR unit (requires (fun h_0 -> is_private h_0 grade /\ None? (sel h_0 grade)) $\label{line:agpre}$
           (ensures (fun h_0 () h_1 -> Some? (sel h_1 grade) /\ $\label{line:agpost:start}$
                               same_labels h_0 h_1 /\ modif_shareable_and !{grade} h_0 h_1)) = $\label{line:agpost:end}$
    let ll = create_llist test in $\label{line:ag-impl:start}$
    label_llist_as_shareable ll; witness (is_shareable ll);
    let res = hw ll in gr := Some (determine_grade res)$\label{line:ag-set-grade}$ $\label{line:ag-impl:end}$
  \end{lstlisting}
\caption{The illustrative example of a verified autograder. Types are simplified.}
\label{fig:autograder_example}
\end{figure}

\subsection{Static verification using Labeled References}\label{sec:key-verification}
In \autoref{fig:autograder_example} we show how to use Labeled References
to verify an autograder.\footnote{
\fstar syntax\ifanon\else~\cite{sciostar}\fi{} is similar to OCaml (\lss{val}, \lss{let}, \lss{match}, etc).
Binding occurrences \lss{b} take the form \lss{x:t}
or \lss{#x:t} for an implicit argument.
%
Lambda abstractions are written
\lss{fun b_1 ... b}$_n$\lss{ -> t},
whereas
\lss{b_1 -> ... -> b}$_n$\lss{ -> C} denotes a curried function
type with result \lss{C}, a computation type describing the effect,
result, and specification of the function.
Contiguous binders of the same type may be written
\lss{(v_1 ... v}$_n$\lss{ : t)}.
Refinement types are written \lss{b\{t\}} 
(\EG \lss{x:int\{x>=0\}} represents natural numbers).
Dependent pairs are written as \lss{x:t_1 & t2}.
%
%
%
For non-dependent function types, we omit the name in the argument binding, \EG type
\lss{#a:Type -> (#m #n : nat) -> vec a m -> vec a n -> vec a (m+n)}
represents the type of the
append function on vectors,
where both unnamed explicit arguments and the return type depend on the
implicit arguments. 
%
%
A type-class constraint \lss{\{| d : c t1 .. tn |\}} is a special
kind of implicit argument, solved by a tactic during elaboration.
%
%
\lss{Type0} is the lowest universe of \fstar{}; we also use it to write propositions,
including \lss{True} (True) and \lss{False} (False).
%
We generally omit universe annotations.
The type \lss{either t_1$~$t_2} has two constructors, \lss{Inl:(#t_1 #t_2:Type) -> t_1 -> either t_1$~$t_2}
and \lss{Inr:(#t_1 #t_2:Type) -> t_2 -> either t_1$~$t_2}.
Expression \lss{Inr? x} tests whether \lss{x} is of the shape \lss{Inr y}.
Binary functions can be made infix by using backticks:
\lss{x `op`$~$y} stands for \lss{op x y}.}
The autograder takes as argument
a function submitted by a student to solve the homework \newtext{(\ls{hw} of type \ls{student_hw})}
and an initially empty (\ls{None}) reference where the grade will be stored (line \ref{line:agtype}).
The function provided by the student is supposed to sort a linked list in place.
The autograder creates a shareable linked list,
executes the homework on it, checks the result,
and sets the grade (lines \ref{line:ag-impl:start}-\ref{line:ag-impl:end}).
The main property verified for the autograder is that the grade stays private and that it gets set to some value (lines \ref{line:agpost:start}-\ref{line:agpost:end}).

Naturally, the homework is unverified since it is written by a student.
Thus, since the autograder calls the unverified homework, to be able to verify it,
one has to assume a specification for the behavior of the homework.
\secrefstar{} then ensures that the homework satisfies the specification.
%
Soundly specifying the behavior of unverified code is one of the main
contributions of this paper, and we explain this below for
\ls{student_hw} (lines $\ref{line:hwtype:start}$-$\ref{line:hwtype:end}$
  in \autoref{fig:autograder_example}), 
illustrating the assumed specification:

\begin{enumerate}[nosep]
  \item The pre-condition requires that the linked list is labeled as shareable,
  statically guaranteeing that the verified autograder passes only shareable references to the homework (line \ref{line:hwpre}).\label{spec:hw-pre}

  \item The post-condition ensures that only shareable references have been modified, 
  and existing references have the same label (line \ref{line:hw-post-modifies}). 
  This is a universal property of all unverified computations,
  which work only with shareable references (\autoref{sec:key-ucode}).
  \label{spec:hw-post-modifies}

  \item The post-condition also ensures that
  if the computation succeeds (by returning \ls{Inl}),
  the linked list is sorted, without cycles, and contains the same values (line \ref{line:hwpost1}).
  This is the part of the specification that gets enforced dynamically
  using higher-order contracts (\autoref{sec:key-hoc}).\label{spec:hw-post-list}
\end{enumerate}

The way we specify the footprint in \autoref{fig:autograder_example} in
the post-condition of the \ls{student_hw} (item (\ref{spec:hw-post-modifies}) above,
  line \ref{line:hw-post-modifies}) is novel.
We specify it as the set of references that are labeled as shareable (using the predicate \ls{modif_only_shareable}).
The usual way of specifying it as the set of references that are
in the linked list
\newtext{(\IE the root reference, the one that follows, and so on to the last reference in the linked list)}
\cite{preorders,Reynolds02,leino10dafny} does not scale in our setting, where we also need to
specify higher-order unverified code. Imagine that the unverified function is in fact
a result of a bigger piece of unverified code and that another liked list was
already shared with it (as in the example from the introduction).
This means that the unverified function can also modify that linked list, thus,
to soundly specify its behavior one has to add to the footprint the references in that linked list too.
Moreover, references that were shared through subsequent updates to the first
linked list also have to be added.
Using our labeling mechanism tames the complexity of globally tracking the
shared references and allows stating the footprint in a simple way.

The assumption that unverified code modifies only shareable references is true if
\begin{inlist}
  \item the unverified code cannot forge references (which we assume),
  \item we treat the references allocated by the unverified code as shareable references
    (which we do and explain in \autoref{sec:key-ucode}), and
  \item the verified code shares only shareable references with the unverified code.
\end{inlist}
To ensure this last point
we added the pre-condition for calling the unverified code that the passed reference has to be shareable
(item (\ref{spec:hw-pre}), line \ref{line:hwpre}).
This pre-condition in combination with a global invariant on the heap enforces
that the explicitly passed reference is shareable, and that
all references that are reachable from the reference are shareable too.
The global invariant on the heap ensures
that a shareable reference can point only to other shareable references.
This means that the \secrefstar{} user has to manually track only
the label of the root reference of a complex data structure (\EG the head of a linked list).
The global invariant is necessary to prevent accidental sharing by subsequent writes to shareable references (\EG the example from the introduction).
The user never has to prove directly the preservation of the global invariant though,
because the operations to manipulate references provided by Labeled References take care of that.
The only extra proof burden appears when labeling a reference as shareable or
writing to a shareable reference, one has to prove some conditions ensuring
that the affected reference will point only to shareable references after the
operation.\ch{please explain this on our examples; otherwise it's quite funny}

The \ls{witnessed} token used in the pre-condition (line \ref{line:hwpre})
is produced by the \ls{witness} operation of Monotonic State,
which provides two such operations:
(1)~\ls{witness}, allowing one to convert stable heap predicates 
into heap-independent ones,
and (2)~\ls{recall}, allowing one to recover an already witnessed predicate
for any future heap.
For example, being shareable is a stable property with respect to our preorder
that shareable references stay shareable. Thus if one knows that a reference
\ls{r} is shareable in a heap \ls{h} (\IE \ls{is_shareable r h}), one can
witness this and get a \newtext{heap-independent} token \ls{witnessed
(is_shareable r)}, which can be recalled later, say when the heap is \ls{h'}, to
prove that \ls{is_shareable r h'}.\footnote{For readers familiar with Separation
Logic, this is similar to invariants and their introduction and elimination
rules.} \newtext{For intuition, making \ls{is_shareable r} into a  heap-independent token \ls{witnessed
(is_shareable r)} can be thought of as defining a proposition, which states that there exists some past heap \ls{h} in which \ls{is_shareable r h} was true, and as \ls{is_shareable r} is stable with respect to any allowed heap changes, \ls{is_shareable r h'} is also true in any future heap \ls{h'} from \ls{h}.}

To be able to statically verify the specifications of the autograder (lines
\ref{line:agpre}-\ref{line:agpost:end}) we reiterate that giving a specification
to the unverified function is necessary.
Since the autograder calls into unverified code, it inherits the property
that it can modify shareable references, thus, its footprint
is defined as the union of all shareable references and
the singleton set of references containing only the grade (line \ref{line:agpost:end}).
The autograder's code stores the grade inside a monotonic reference
\newtext{(denoted by \ls{mref} on line \ref{line:agtype}), a feature of Monotonic State~\cite{preorders}
that allows putting a statically enforced preorder on the value stored in a reference.}
%
For the grade, the preorder ensures that the grade cannot be modified once it is set.
To be able to set the grade (on line \ref{line:ag-set-grade}), however, one has to know that the grade was not set by
calling into the homework, which we know because
the footprint of the call states that it modifies only shareable references, and we know
that the grade was private before the call. \newtext{The preorder on the grade is defined as follows:}
\begin{lstlisting}
let grade_preorder (g_0 g_1:option int):Type0=match g_0 with |None->True$\ $|Some v_0->Some? g_1/\v_0=Some?.v g_1
\end{lstlisting}

In \secrefstar{}, one cannot share a monotonic reference with the unverified code because
there would be no way to guarantee the preorder of the reference, but the verified code
can use monotonic references.
One can share only normal references, denoted by \ls{ref}, as long as they are labeled as shareable.
Normal references are just monotonic references with the trivial preorder.







\subsection{Representation for unverified code}\label{sec:key-ucode}

We give a representation for unverified code as a shallow embedding in \fstar{} satisfying two constraints:
first, obviously, it has to be a representation for {\em unverified} code
(\IE it requires only trivial logical reasoning),
and second, we need to be able to show that any piece of unverified code using this representation satisfies
the universal property that it modifies only shareable references.

We found a representation that satisfies both these constraints and that
is also unaware of the labeling mechanism, by using the effect \ls{MST}.
The representation uses abstract predicates as well as a version of the
reference-manipulating operations that mention only these abstract predicates in their specifications.
The intuition why this is an appropriate representation for {\em unverified} code in \fstar{}
is that one can sequence the new operations freely, because the operations
require and ensure only the abstract predicates, so one only has to carry the abstract predicates around
from post-conditions to pre-conditions, but one never has to do nontrivial logical reasoning.
\newtext{To make it even more apparent that it is an appropriate representation,
in the artifact we also have a syntactic representation,
a small deeply embedded \(\lambda\)-calculus with first-order ML-style references,
and we show a total (back)translation function in \fstar{} from any well-typed
syntactic \(\lambda\)-calculus expression to the shallowly embedded representation with abstract predicates.}

The type of the new reference-manipulating operations is presented in \autoref{fig:type-poperations}. \newremove{and
they}\newtext{They} use three abstract predicates in a way that allows the operations to still be freely sequenced:
\begin{inlist}
  \item a global invariant on the heap (\newremove{denoted by }$\inv$), required and guaranteed by all operations;
  \item a predicate on references (\newremove{denoted by }$\prref$) 
        constraining references stored in the heap by \ls{write} and \ls{alloc},
        \newtext{ensured of freshly allocated references using \ls{alloc},
        and ensured of references returned from the heap by \ls{read}};
  \item a preorder on heaps (\newremove{denoted by }$\hrel$) and a guarantee that for all our
        operations the heap evolves according to this preorder.
\end{inlist}
In \autoref{fig:type-poperations} we also use
the $\every$ combinator that takes a predicate on references
(in our case $\prref$) and a value of an inductive type,
then recursively decomposes the value by pattern matching.
For each reference, it applies the predicate, ensuring that all
reachable references satisfy \newremove{the given predicate.}\newtext{it.}

\begin{figure}
    \begin{lstlisting}
type poly_alloc_t ($\inv$:heap -> Type0) ($\prref$:ref 'a -> Type0) ($\hrel$:preorder heap) =
    init:'a -> MST (ref 'a) (fun h_0 -> $\inv$ h_0 /\ $\every$ $\prref$ init) (fun h_0 r h_1 -> h_0 $\hrel$ h_1 /\ $\inv$ h_1 /\ $\prref$ r)  
type poly_read_t ($\inv$:heap -> Type0) ($\prref$:ref 'a -> Type0) ($\hrel$:preorder heap) = 
    r:ref 'a -> MST 'a$\ $(requires (fun h_0 -> $\inv$ h_0 /\ $\prref$ r)) (ensures (fun h_0 v h_1 -> h_0 $\hrel$ h_1 /\ $\inv$ h_1 /\ $\every$ $\prref$ v))
type poly_write_t ($\inv$:heap -> Type0) ($\prref$:ref 'a -> Type0) ($\hrel$:preorder heap) =
    r:ref 'a -> v:'a -> MST unit (fun h_0 -> $\inv$ h_0 /\ $\prref$ r /\ $\every$ $\prref$ v) (fun h_0 _ h_1 -> h_0 $\hrel$ h_1 /\ $\inv$ h_1)
\end{lstlisting}
\caption{Types of the reference-manipulating operations polymorphic in $\inv$, $\prref$ and $\hrel$.
\newtext{The $\every$ combinator is defined using a type class and here it is elided that type \ls{'a} is constrained
by this type class.}}
\label{fig:type-poperations}
\end{figure}

\newtext{We call this representation, a \textit{polymorphic interface},}
and we used it for instance to write the
following unverified homework that sorts a linked list, which
takes as argument the three operations and combines them freely.\footnote{We
  elided \newtext{the use of fuel to have a} proof of termination, which is needed
  since for simplicity our work assumes that state is the only
  side effect. \newtext{The full implementation is in the artifact.}}

%
\begin{lstlisting}
type poly_student_hw = $\inv$:(heap -> Type0) -> $\prref$:(ref 'a -> Type0) -> $\hrel$:preorder heap ->
  (* set of operations that use the abstract predicates *)
  alloc : poly_alloc_t $\threep$ -> read : poly_read_t $\threep$ -> write: poly_write_t $\threep$ ->
  (* the actual type of the homework *)
  ll:ref (linked_list int) -> 
  MST (either unit err) (requires (fun h_0 -> $\inv$ h_0 /\ $\prref$ ll)) (ensures (fun h_0 r h_1->h_0$\hrel$h_1 /\ $\inv$ h_1 /\ $\every$ $\prref$ r))
let unverified_student_hw : poly_student_hw = fun my_alloc my_read my_write ->
  let rec sort (ll:ref (linked_list int)) =
    match my_read ll with | LLNil -> ()
    | LLCons x tl -> begin
      sort tl;
      match my_read tl with | LLNil -> ()
      | LLCons x' tl' -> if x <= x' then () else (
        let new_tl = my_alloc (LLCons x tl') in
        sort new_tl;
        my_write ll (LLCons x' new_tl))
    end
  in sort
\end{lstlisting}
The abstract predicates are also used for specifying the homework in the target 
language,\ch{+on line 6 above}, and we explain how \fstar{} shows that the 
homework satisfies its post-condition, even if it is abstract.
The first conjunct of the post-condition
is about the preorder on heaps $\hrel$, which is satisfied because all our 3 operations ensure
that the heap evolves according to $\hrel$,
which is a reflexive and transitive relation.
The second conjunct about global invariant $\inv$ is satisfied because it is
required in the pre-condition of \ls{poly_student_hw}, and then every
operation requires and ensures this invariant.
The third and last conjunct states that any reference that is returned has to satisfy predicate $\prref$.
This conjunct is satisfied because a reference can be allocated (which ensures $\prref$),
read from the heap (which ensures $\prref$),
or is received from the outside, which is also guaranteed because of the pre-condition
requiring that the input linked list satisfies $\prref$.
The pre-condition and the post-condition of \ls{poly_student_hw} ensure
that every reference that crosses the verified-unverified boundary satisfies the predicate $\prref$.

\newremove{
We therefore claim that this is an appropriate representation for stateful unverified code
shallowly embedded in \fstar{}.
To make it even more apparent, in the artifact we also have a syntactic representation,
a small deeply embedded lambda calculus with first-order ML-style references,
and we show a total (back)translation function in \fstar{} from any well-typed
syntactic lambda-calculus expression to the shallowly embedded representation with abstract predicates.}


%

Now time for some magic: we show that we can instantiate the \newremove{types}\newtext{polymorphic interface} we give to unverified
code, so that this code gets closer to something that can be safely linked with verified code.
We illustrate this instantiation step by step for \ls{poly_student_hw} above.
Instantiating the abstract invariant, $\inv$, lifts all the functions to the effect \ls{LR},
which is just a synonym for the effect \ls{MST} with an extra global invariant,
ensuring that shareable points to shareable,
so after instantiation we have the same effect
as in the type \ls{student_hw} from \autoref{fig:autograder_example}.
We instantiate 
$\prref$ with \ls{fun r -> witnessed (is_shareable r)},
which becomes the same pre-condition as in type \ls{student_hw}.
Then, we instantiate the preorder 
$\hrel$, with \ls{fun h_0 h_1 -> modif_only_shareable h_0 h_1 /\\ same_labels h_0 h_1},
which means we now have as post-condition the universal property
that unverified code modifies only shareable references,
which is also the first part of the post-condition of \ls{student_hw}.
So unverified code is instantiated during linking to obtain a type with concrete pre- and post-conditions,
\newtext{now having what we call an \textit{intermediate interface}.}
%
%
This is a first step towards bridging the gap between verified and unverified
code \newtext{(see \autoref{fig:overview})},
and in \autoref{sec:key-hoc} we show how we additionally use higher-order
contracts during the compilation of the verified program to fully bridge this gap.

But first, let us still mention that instantiating the three operations from
\autoref{fig:type-poperations} is straightforward.
We instantiate allocation by allocating a reference and then labeling it as shareable. 
%
And we instantiate \ls{read} and \ls{write} using the default operations of Labeled References,
which because of the specs given by the three concrete predicates,
they are now specialized for shareable references.
For \ls{read} an interesting interaction happens between $\every\ \prref\ $\ls{r} and
the global invariant to show that whatever was read from the heap is also shareable.
The $\every$ combinator stops when reaching \newtext{any} reference
because once it was established that a reference is shareable, the global invariant offers the
guarantee that all references reachable from this reference are also shareable.

\begin{figure}[h]
  \begin{tikzpicture}[auto, node distance=5cm]
  \draw
      node at (0, 0) {}
      node [block, rounded corners] (vps) {\( \small \begin{array}{cc}\textit{verified}\\\textit{program}\end{array} \)}
      node [below of=vps, node distance=0.8cm, align=center] (vps1) {\tiny strong\\[-1.5ex]\tiny interface}
      node [block, right of =vps, rounded corners] (vpi) {\( \small \begin{array}{cc}\textit{compiled}\\\textit{program}\end{array} \)}
      node [below of=vpi, node distance=0.8cm, align=center] (vpi1) {\tiny intermediate\\[-1.5ex]\tiny interface}
      node [block, right of =vpi, rounded corners] (upt) {\( \small \begin{array}{cc}\textit{unverified}\\\textit{code}\end{array} \)}
      node [below of=upt, node distance=0.8cm, align=center] (upt1) {\tiny polymorphic\\[-1.5ex]\tiny interface}
      ;
  \draw[-] [thin] (1.25, 0.75) -- (6.45, 0.75);
  \draw[-] [thin] (8.40, 0.75) -- (10.75, 0.75);
  \draw node at (7.425, 0.75) {\tiny target language};
  \draw node at (0.0, 0.75) {\tiny Labeled References};
  \draw[->](vps) -- node [above] { \small {\bf compile} } node [below,align=center] { \small add higher-\\[-1ex]\small order contracts } (vpi);
  \draw[to reversed-to reversed](vpi) -- node [above] { \small {\bf link} } node [below,align=center] { \small instantiate\\[-1ex]\small abstract predicates } (upt);
  \draw[-] [color=gray,thick, dotted] (1,-0.9) -- (1,0.75);
  \end{tikzpicture}
  \caption{\label{fig:overview}An overview of \secrefstar{}.\vspace{-1em}}
\end{figure}

\subsection{Higher-order contracts in \titsecrefstar}\label{sec:key-hoc}
After the instantiation, an unverified computation of type \ls{unverified_student_hw}
almost has the type\newtext{, which we call a \textit{strong interface},} \ls{student_hw}.
The only missing piece is a part of the post-condition (\autoref{spec:hw-post-list} in \autoref{sec:key-verification}),
which states that after the homework returns, if the run was successful, then the linked list
is sorted, has no cycles and contains the same values as before.
This part is enforced by \secrefstar{} using higher-order contracts~\cite{FindlerF02},
which are used to enforce only properties of shareable references or
values passed when unverified code gives control to the verified code.
Since the verified code has no insights into what the unverified code does with the shareable references
and with which values it gives them, no logical predicates can be established about them.
Therefore, it is the choice of the user if they want to strengthen the types and use higher-order
contracts, or if they want to manually establish the logical predicates by adding dynamic checks. The advantage of 
strengthening the specification is that if one would instead link with some verified program, then there would be no need for the dynamic checks,
and for unverified code, one would use \secrefstar{}.

We implement the idea of Stateful Contracts~\cite{FindlerF02} using 
wrapping and by implementing
two dual functions \ls{export} and \ls{import}: \ls{export} takes existing specs and converts them into dynamic checks
(\EG converts pre-conditions of arrows and refinements on the arguments to dynamic checks),
while \ls{import} adds dynamic checks to have more specs
(\EG adds dynamic checks to strengthen the post-conditions of arrows and to add refinements on the result).\ch{
  The parens make some sense to me, I wasn't able to follow the rest of the explanation!}
The functions \ls{export} and \ls{import} are defined in terms of each other to support higher-order functions.
In \secrefstar{}, when the verified program has initial control,
the \ls{import} function is used during both compilation and back-translation
(the dual of compilation, used for the secure compilation proof in \autoref{sec:secure-compilation}).
During compilation \ls{import} is used to add dynamic checks
to the verified-unverified boundary of the compiled program to protect it. During back-translation
\ls{import} is used to add dynamic checks to the unverified code so that it can be typed
as having the necessary specification. In the dual setting, when
the unverified code has the initial control, the \ls{export} function is used.\ch{used on/for what?}

Our implementation of higher-order contracts is verified not to modify the heap in any way
(it only reads from the heap).
%
%
\secrefstar{} also verifies if the
dynamic checks imply the specification that has to be enforced.
Since \fstar{} specifications are not directly executable,
\secrefstar{} leverages type classes to automatically infer
the check. If this fails though, the user is asked to manually
provide the check and prove that the check implies the specification,
which benefits from SMT automation.

\subsection{Encapsulated references}\label{sec:key-encapsulated}

We want to enable the verified program to mutate private references inside the callbacks
it passes to the unverified code.
Such references are not shared with the unverified code, but are modified indirectly by
the callback calls the unverified code makes.
To enable this use case, we introduce a third label, \ls{encapsulated}.
Only references labeled as private can become encapsulated,
and they remain encapsulated forever.
Since encapsulated references are used only in verified code, there is no need to
add a constraint that limits to which kind of references they point.

Encapsulated references are only directly modified by the verified program, but an unknown number of times,
since the unverified code can call back as many times as it likes.
The callback can also be called at later times, since it can be stashed away.
The advantage of encapsulated references is that one knows {\em how} they get modified,
thus, one can encapsulate {\em monotonic} references to take advantage of a preorder, and/or use refinement types
and/or dependent types for the value stored.

Encapsulated references are especially
useful when the unverified code has initial control. The unverified code can call multiple times
the verified program, but in a first-order setting the calls are independent, and
any reference allocated by the program gets lost when it gives control
back to the unverified code. 
If the program is higher-order though, it can have callbacks with
persistent state by using encapsulated references.
We illustrate this case in \autoref{fig:pngr_example} using the example of a pseudo-number generator
that has an initialization phase where it gets a seed as an argument and returns a callback
that generates a number on each call. To generate a new number for every call to the callback,
the program allocates an encapsulated reference, \ls{counter}, with which it keeps track of how many times it has been called.
This is a monotonic reference with the preorder that it can only increase.
Since the generator does not make any assumption about the unverified code (the pre-conditions are trivial),
\secrefstar{} does not add any dynamic checks during compilation.

\begin{figure}[h]
  \begin{lstlisting}[numbers=left,escapechar=\$,xleftmargin=13pt,framexleftmargin=10pt,numberstyle=\color{dkgray}]
val generate_nr : int -> int -> int (* a pure computation *)
let post = fun h_0 _ h_1 -> modif_shareable_and_encaps h_0 h_1 /\ same_labels h_0 h_1
let prng (seed:int) 
  : LR (unit -> LR int (requires (fun _ -> True))$\ $(ensures post)) (requires (fun _ -> True))$\ $(ensures post) =
    let counter : mref int (<=$\!$) = alloc 0 in label_encapsulated counter;
    (fun () -> counter := !counter + 1; generate_nr seed (!counter))
\end{lstlisting}
\caption{The illustrative example of a pseudo-number generator}
\label{fig:pngr_example}
\end{figure}

This last example also reveals that the presentation in \autoref{fig:autograder_example}
and \autoref{sec:key-ucode} is simplified to improve readability, and in fact, \secrefstar{}
uses the more general predicate \ls{modif_shareable_and_encaps}
(instead of the more specific \ls{modif_only_shareable}) as a post-condition for unverified code.

\section{Monotonic State represented using a Dijkstra Monad}
\label{sec:monotonic-state}

We first introduce the Monotonic State monadic effect \ls{MST} on which Labeled References builds.
We present the interface of Monotonic State as in the original work 
of \citet{preorders},
after which we present our contribution involving defining \ls{MST} as a Dijkstra
monad with a free-monad-based representation (\autoref{sec:mst-representation}),
and finally we present a mechanized soundness proof for \ls{MST} (\autoref{sec:mst-soundness}).

\subsection{The interface of Monotonic State}
\label{sec:mst-high-level}
Monotonic State is an \fstar effect characterizing stateful computations whose
state evolves according to a given preorder, with additional features to make
specifying and verifying such computations more convenient. In particular, given
some predicate on states that is stable for any state changes respecting the
preorder, one can witness that the predicate \newtext{\ls{p}} holds in one part
of a program, resulting in a state-independent token \newtext{\ls{witnessed p}}
witnessing this action, and then use this token to recall that the predicate
\newtext{\ls{p}} still holds at arbitrary later program points \newtext{and
in later states}.

\citet{preorders} considered a general \newtext{notion of} Monotonic State\newremove{notion}, in which
the effect is parameterized by an arbitrary type of states and a
preorder on them.
In this paper, we pick as states monotonic heaps (type \ls{heap}) provided by \fstar{},
with the preorder $\preceq$ \newtext{on} heaps enforcing that allocated references stay allocated forever,
and monotonic references preserve their preorders.
With this, we obtain our \ls{MST} effect, which models
first-order\footnote{ Providing a monadic representation that allows storing
  effectful functions given \fstar{}'s predicative, countable hierarchy of
  universes is an open research problem, further detailed in
  \autoref{sec:mst-representation} and \autoref{sec:future}.}
ML-style mutable references.

The interface of \ls{MST}, presented in \autoref{fig:monotonic-state},
hides the concrete representation of the
heap and references, and the implementation of the
operations that manipulate them.
The interface provides proof that the
operations preserve the preorder
through a series of lemmas (which we elide).


The basic reference type is that of monotonic references, \ls{mref a rel}, which
have individual preorders enforced by the preorder on the heaps $\preceq$.
ML-style references \ls{ref a} are a special
case for the \newremove{trivial}\newtext{total} preorder that relates all the \ls{a}-values.
Behind the interface, the references have concrete addresses which are hidden
from the user to ensure their unforgeability.
For reasoning, however, the interface exposes a function \ls{addr_of} that
returns an \ls{erased}, computationally irrelevant
value that can only be used in specifications, but not branched on in programs,
ensuring the abstraction boundary is respected.
Labeled References uses \ls{addr_of} in the encoding of the labeling mechanism.

The \ls{MST} effect comes with operations for reading from, writing to, and allocating fresh references.
The pre-condition of \ls{write} requires the new value to be
related to the value of the reference in the \newremove{previous}\newtext{current} heap. This
ensures that the reference evolves according to the preorder of the reference. The
\ls{write} and \ls{alloc} operations both contain a \ls{modifies} clause in
their post-conditions, which is the usual way of specifying that only specific
references are allowed to be modified (the footprint of the call\newremove{, no labels yet}).
Specifically, for a set \ls{s} of
addresses, \ls{modifies s h_0 h_1} holds if any references contained in \ls{h_0}
and not contained in \ls{s} remain unmodified. Above, \ls{!\{r\}} is \fstar's
syntax for the singleton set containing the address of \ls{r}. Notice that the
pre-conditions of \ls{read} and \ls{write} require the reference to be contained,
yet for simplicity, we often elide this detail from the other sections of the paper.

\newtext{For users of \ls{MST}, it is important to note that both the original
work of \citet{preorders} and our work on \ls{MST} rely for their soundness on
the logical \ls{witnessed} tokens being abstract, unforgeable, and only
constructible using the \ls{witness} operation of the \ls{MST} effect.
Therefore, one must avoid using strong abstraction and parametricity breaking
classical axioms in code written using the \ls{MST} effect, such as the axiom of
strong excluded middle that equates \ls{prop} and \ls{bool} and allows one to
test the validity of arbitrary propositions, because that leads to an
inconsistency in \fstar{}.\footnote{Monotonic State is Incompatible with Strong
Excluded Middle:
\href{https://github.com/FStarLang/FStar/issues/2814}{https://github.com/FStarLang/FStar/issues/2814}}
Avoiding such axioms in user code is standard practice in \fstar-based projects
because of their abstraction- and parametricity-breaking nature. They are
typically used very carefully only by library designers.}

\begin{figure}
\begin{lstlisting}
val heap : Type0
val mref: a:Type0 -> rel:preorder a -> Type0
type ref (a:Type0) = mref a (trivial_preorder a)
val contains: #a:Type0 -> #rel:preorder a -> heap -> mref a rel -> Type0
val sel: #a:Type0 -> #rel:preorder a -> h:heap -> r:mref a rel{h `contains` r} -> a
val heap_write: #a:_ -> #rel:_ -> h:heap -> r:mref a rel{h `contains` r} -> x:a{sel h r `rel` x} -> heap
val heap_alloc: #a:Type0 -> rel:preorder a -> heap -> a -> mref a rel & heap
val addr_of: #a:Type0 -> #rel:preorder a -> mref a rel -> erased pos
val next_addr: heap -> erased pos
let $(\preceq)$ h_0 h_1 = forall a rel (r:mref a rel). h0 `contains` r ==> (h1 `contains` r /\ (sel h0 r) `rel` (sel h1 r))
val alloc: #a:Type -> #rel:preorder a -> init:a -> 
  MST (mref a rel) (requires (fun h_0 -> True)) 
                  (ensures (fun h_0 r h_1 -> h_0 $\preceq$ h_1 /\ fresh r h_0 h_1 /\ modifies !{} h_0 h_1 /\ 
                                      addr_of r == next_addr h_0))
val read: #a:_ -> #rel:_ -> r:mref a rel -> MST a (requires (fun h_0 -> h_0 `contains` r)) 
                                           (ensures (fun h_0 v h_1 -> h_0 == h_1 /\ v == sel h_1 r))
val write: #a:_ -> #rel:_ -> r:mref a rel -> v:a -> 
  MST unit (requires (fun h_0 -> h_0 `contains` r /\ (sel h_0 r) `rel` v))
           (ensures (fun h_0 () h_1 -> h_0 $\preceq$ h_1 /\ modifies !{r} h_0 h_1 /\ equal_dom h_0 h_1))
val witnessed: p:(heap -> Type0) -> Type0
let stable_heap_predicate = p:(heap -> Type0){forall h_0 h_1. p h_0 /\ h_0 $\preceq$ h_1 ==> p h_1}
val witness: p:stable_heap_predicate -> MST unit (requires (fun h_0 -> p h_0)) 
                                              (ensures (fun h_0 () h_1 -> h_0 == h_1 /\ witnessed p))
val recall: p:stable_heap_predicate -> MST unit (requires (fun _ -> witnessed p)) 
                                            (ensures (fun h_0 () h_1 -> h_0 == h_1 /\ p h_1))                  
\end{lstlisting}
\caption{Interface of Monotonic State}\label{fig:monotonic-state}
\end{figure}

\subsection{\texorpdfstring{\ls{MST}}{MST} as a Dijkstra monad with a free-monad-based representation}
\label{sec:mst-representation}

\newcommand{\D}{\mathcal{D}} 

In the work of \citet{preorders}, the \ls{MST} effect and the operations
\newtext{discussed} above were simply assumed \newtext{in \fstar{} (using the
\ls{assume} keyword)} and not given \newtext{any representation or} definitions
within \fstar. \newremove{We}\newtext{In this paper, we} improve on this \newtext{situation} by giving a free-monad-based
\newremove{representation}\newtext{\emph{representation}} for \ls{MST} in
\fstar{}, from which it can be defined as a Dijkstra monad \cite{dm4all}.
\newtext{While the main definition of the \ls{MST} effect remains axiomatic, via
the monad morphism \lstinline{theta} we define below,\footnote{\newtext{The axiomatic
nature of the effect definition is due to the use of a free monad as its
representation, whose elements are inert trees and where no actual computation
takes place. As such, the specifications assigned by the monad morphism play the
same role as the assumed types of operations in the existing work, but now we also
have a representation to work with.}} the existence of a free-monad-based
representation enables us to separately define a computational state-passing semantics for \ls{MST} computations, and use that semantics to for the first time prove the soundness of the \ls{MST} effect interface
within \fstar{} itself (\autoref{sec:mst-soundness}), whereas \citet{preorders}
gave their soundness proof only on paper because there was no
representation of \ls{MST} computations to work with inside \fstar.}

We recall that a Dijkstra monad is a
monad-like structure indexed by specifications, which themselves often
come from some monad, called a specification monad. For a Dijkstra monad $\D$,
the type $\D\, \alpha\, w$ classifies computations that return $\alpha$-typed values
and satisfy a specification $w : W\, \alpha$, where $W$ is an
aforementioned specification monad. A common way to define a
Dijkstra monad $\D$ \cite{dm4all} is to identify (i) a monad $M$ to represent
computations classified by $\D$, (ii) an ordered monad $W$ that models
specifications for the computations in $M$, where the order formalizes when one
specification implies another, and (iii) a monad morphism $\theta : M \to W$
that infers for every computation its ``most precise''
specification. The Dijkstra monad $\D$ is then defined as $\D\, \alpha\, w =
m:M\, \alpha\{ \theta(m) \sqsubseteq w \}$, formalizing the idea that $\D\,
\alpha\, w$ classifies computations $m : M\, \alpha$ that satisfy the given
specification $w$. Below we work out such a Dijkstra monad for \ls{MST}.

\paragraph{The computational monad.}
We begin by defining the computational monad as a free monad:\footnote{We omit a
constructor that is needed in \fstar to account for partiality coming from
pre-conditions \cite{pdm4all}.}
\begin{lstlisting}
type mst_repr (a:Type u#a) : Type u#(max 1 a) =
| Return : a -> mst_repr a  
| Read : #b:Type0 -> #rel: preorder b -> r: mref b rel -> k:(b -> mst_repr a) -> mst_repr a
| Write : #b:Type0 -> #rel: preorder b -> r: mref b rel -> v:b -> k:mst_repr a -> mst_repr a
| Alloc : #b:Type0 -> #rel: preorder b -> init: b -> k:(mref b rel -> mst_repr a) -> mst_repr a
| Witness : p:stable_heap_predicate -> k:(unit -> mst_repr a) -> mst_repr a
| Recall : p:stable_heap_predicate -> k:(unit -> mst_repr a) -> mst_repr a
\end{lstlisting}
As is standard with free monads, elements of \ls{mst_repr} denote computations
modelled as trees whose nodes are operation calls and whose leaves contain
return values. Observe that the continuations of \ls{Witness} and
\ls{Recall} accept simply \ls{unit}-typed values---at the level of the
computational representation no \ls{witnessed p} tokens are passed to the
continuation, they only appear and are passed around in the logical
specifications. The return and bind
operations of this free monad are standard.

Let us also discuss why we cannot store higher-order effectful functions.
Notice that our monad transforms a \ls{Type u#a} into a \ls{Type u#(max 1 a)},
this is because the operations \ls{Read}, \ls{Write}, and \ls{Alloc} quantify
over \ls{Type0}. Thus, it increases the universe level,
which makes functions represented using this monad unstorable in our references.
There is no known monadic representation that allows storing effectful
functions given \fstar{}’s predicative, countable hierarchy of universes.
We support references that point to other references but cannot point to effectful functions.

 
\paragraph{The specification monad}we use, \ls{mst_w}, is
the monad of weakest pre-conditions:
\begin{lstlisting}
let mst_w (a:Type) = wp:(post:(a -> heap -> Type0) -> pre:(heap -> Type0))
  {forall post post'. (forall v h_1. post v h_1 ==> post' v h_1) ==> (forall h_0. wp post h_0==>wp post' h_0)}
\end{lstlisting}
A specification\newremove{\ls{wp:mst_w a}} is a function that transforms any
post-condition\newremove{\ls{post:a -> heap -> Type0}} over return values
and final heaps into a
pre-condition\newremove{\ls{wp post:heap -> Type0}} over initial heaps.
The refinement type
restricts us to monotonic predicate transformers. The unit of this monad
requires the post-condition to hold of the initial state, as \ls{mst_w_return v
post h_0 = post v h_0}, and the bind operation composes two predicate
transformers, as \ls{mst_w_bind wp wp' post h_0 = wp (fun v h_1 -> wp' v post
h_1) h_0}. The order \ls{[=} on \ls{mst_w} is defined so that \lstinline{wp [=$\xspace$ wp'}
holds when \ls{forall post h_0 . wp' post h_0 ==> wp post h_0}.

\paragraph{The monad morphism and the Dijkstra monad.}
As the final ingredient for defining a Dijkstra monad, we define a monad
morphism \lstinline{theta: mst_repr 'a -> mst_w 'a} by recursion on the tree
structure of \ls{mst_repr 'a}\newtext{, which assigns 
weakest-pre-condition-based specifications to \ls{MST} operations and 
computations}. For illustration,
the cases for \ls{return}, \ls{witness}, and \ls{recall} are given as follows.
\begin{lstlisting}
let rec theta m = match m with | Return v -> mst_w_return v | Read r k -> ... | Write r v k -> ... | Alloc init k -> ... 
| Witness p k -> mst_w_bind (fun post h -> p h /\ (witnessed p ==> post () h)) (fun v -> theta (k r))
| Recall p k -> mst_w_bind (fun post h -> witnessed p /\ (p h ==> post () h)) (fun v -> theta (k r))
\end{lstlisting}  
The unit of \ls{mst_repr} is mapped to the unit of \ls{mst_w}, and for the
operations the bind of \ls{mst_w} is used to compose the specificational
interpretation of operations with the recursive application of \ls{theta} to
their continuations. For instance, the \ls{mst_w} interpretation of
\ls{Witness} matches exactly the specification of \ls{witness} from
\autoref{sec:mst-high-level} using the correspondence between pre- and
post-conditions and predicate transformers \cite{dm4all}. The same pattern repeats
for all the other operations. 

The Dijkstra monad is defined by refining \ls{mst_repr} computations using \ls{mst_w} predicates via \ls{theta}:
\begin{lstlisting}
type mst_dm (a:Type) (wp:mst_w a) = m:mst_repr a{theta m [=$\xspace$ wp}
\end{lstlisting}

\paragraph{The layered effects.}
So far \ls{mst_dm} is just an ordinary type in \fstar. In order to turn it into
the \ls{MST} effect discussed earlier to support direct-style programming (where
monadic returns and binds are implicit), we use a mechanism \newremove{to
extend}\newtext{for extending} \fstar with new effects \cite{indexedeffects}. 


\begin{lstlisting}
effect MSTATE (a:Type) (w:mst_w a) with { repr = mst_dm; return = mst_dm_return; ... }
effect MST (a:Type) (pre:heap -> Type0) (post:heap -> a -> heap -> Type0) = 
  MSTATE a (fun p h_0 -> pre h_0 /\ (forall v h_1 . h_0 $\hrel$ h_1 /\ post h_0 v h_1 ==> p v h_1))
\end{lstlisting}
The last effect, \ls{MST}, is derived as a pre- and post-condition based effect synonym~\cite{dm4all},
and is the one we
discussed in \autoref{sec:mst-high-level}. In \fstar we mark \ls{MSTATE} (and
thus also the derived effects) \ls{total}, \ls{reifiable}, and \ls{reflectable}.
Since our monadic representation is based on an inductive type,
we indicate to \fstar{} that our effect is total.
Totality means that any programs that typecheck against these effects
are guaranteed to terminate, and vice versa, any programs we write in these
effects have to be proven to terminate. Reification allows us to reveal the
underlying free-monad-based representation of these effects in proofs, and
reflection allows us to reflect the free monad operation calls to operations
in these effects, with the specification of the reflected code determined by
\ls{theta}.
%

\subsection{Soundness of the free-monad-based Dijkstra monad for Monotonic State}\label{sec:mst-soundness}

\iffull
\else
\newremove{In}\newtext{As noted earlier, in} the work of \citet{preorders}
\fi
the \ls{MST} effect \newremove{was simply axiomatized}\newtext{interface was
simply assumed} in \fstar, with its computational interpretation defined and the
soundness of its axiomatization proved\newremove{, both} \newtext{only} on
paper. \newtext{Also, the proof of \citet{preorders} only applied to a $\lambda$-calculus much less powerful than full \fstar.} \newremove{Above we defined such a computational interpretation in \fstar
and now give a soundness proof that the specifications computed with \ls{theta}
from the free monad representation are computationally satisfied.} \newtext{To improve on this, we now use the free-monad-based representation of \ls{MST} to define a
computational state-passing semantics for \ls{MST} and use it to give a
soundness proof in \fstar that the specifications computed with \ls{theta} from the free
monad representation are computationally satisfied.}
A soundness proof is desirable because
choosing a specification for the free monad operations in \ls{theta} is not
necessarily the same as being sound with respect to the intended computational behavior.


At the core of our soundness proof is the idea that the \ls{witnessed} tokens
are merely a convenience for communicating logical information from one part of
a program to another one, and any program that is provable using \ls{witnessed}
tokens ought to also be provable without them.
%
To be able to vary the relevance of \ls{witnessed} tokens, we consider a variant of \ls{theta}, written 
\lstinline{theta}$_{\!\textsf{w}}$, 
which is parametric in the tokens \ls{w : (heap -> Type0) -> Type0} 
used in the specifications of \ls{Witness} and \ls{Recall}, but is otherwise defined like \ls{theta}.
We also define a predicate \ls{witnessed_before} that captures which
stable heap predicates are recalled in a computation, but were not witnessed in the same computation.
In particular, when we can prove \ls{empty `witnessed_before` m}, then any
\ls{witnessed p} token recalled in \ls{m:mst_repr 'a} had to originate from an earlier call
to \ls{witness p} in \ls{m}, which is a property of closed programs.
\begin{lstlisting}
let rec witnessed_before (preds:set stable_heap_predicate) (m:mst_repr 'a) : Type0 =
  match m with | Return _ -> True | Read r k -> ... | Write r v k -> ...
  | Alloc init k -> forall r. preds `witnessed_before` (k r) 
  | Witness pred k -> (union preds (singleton pred)) `witnessed_before` (k ())
  | Recall pred k -> pred `mem` preds /\ preds `witnessed_before` (k ())
\end{lstlisting}  

We define the state-passing semantics \ls{run_mst} for a closed program \ls{m}.
We refine \ls{m} as being verified to satisfy a specification \ls{wp},
but we instantiate the \ls{w} parameter of
\ls{theta}$_{\!\textsf{w}}$ with the trivial tokens 
\ls{witnessed_trivial p = True}, 
modelling that witnessing and recalling are computationally irrelevant.
Using refinement types, we require that the weakest pre-condition \ls{wp post} holds of the initial heap,
and we ensure that the post-condition \ls{post} holds of the return value and the final heap.
During the recursive traversal of \ls{m},
the heap is augmented with a set of predicates witnessed during execution
and with a refinement type ensuring that these predicates hold for the given heap.
The refinement holds as the heap evolves, since the operations that manipulate the heap
preserve the preorder \newremove{for} which the witnessed predicates are stable \newtext{for}.
During the traversal, the refinement with \ls{witnessed_before} ensures
that any predicates witnessed are contained in the set of predicates of the augmented heap.
\begin{lstlisting}
type heap_w_preds = h:heap & preds:(set stable_heap_predicate){forall pred. pred `mem` preds ==> pred h}
let rec run_mst_with_preds (wp:state_wp 'a) (m:mst_repr 'a{theta$_{\!\textsf{witnessed\_trivial}}$ m [= wp}) 
  (post:('a -> heap -> Type0)) (hp_0:heap_w_preds{wp post (fst hp_0) /\ (snd hp_0) `witnessed_before` m}) 
: (r:'a & hp_1:heap_w_preds{post r (fst hp_1)}) 
= match m with | Return v -> (v, hp_0) | Read r k -> ... | Write r v k -> ...
  | Alloc #b #rel init k -> let (r,h_1) = heap_alloc rel (fst hp_0) init in
      run_mst_with_preds (theta$_{\!\textsf{witnessed\_trivial}}$ (k r)) (k r) post (h_1,snd hp_0)
  | Witness p k -> let lp = (snd hp_0) `union` (singleton p) in
      run_mst_with_preds (theta$_{\!\textsf{witnessed\_trivial}}$ (k ())) (k ()) post (fst hp_0, lp)
  | Recall p k -> run_mst_with_preds (theta$_{\!\textsf{witnessed\_trivial}}$ (k ())) (k ()) post hp_0
let run_mst (wp:state_wp 'a) (m:mst_repr 'a{theta$_{\!\textsf{witnessed\_trivial}}$ m [= wp /\ empty `witnessed_before` m}) 
  (post:('a -> heap -> Type0)) : h_0:heap{wp post h_0} -> r:'a & h_1:heap{post r h_1}
= fun h_0 -> let (r, hp_1) = run_mst_with_preds wp m post (h_0, empty) in (r, fst hp_1)
\end{lstlisting}

Using \ls{run_mst}, we define the soundness theorem
which states that if a closed program \ls{m} is well-typed against a specification that uses
abstract \ls{witnessed} tokens (modelled using universal quantification), then for any
post-condition \ls{post} and initial heap \ls{h_0}, if the weakest
pre-condition \ls{wp post} holds in \ls{h_0}, the post-condition holds for
the return value \ls{v} in the final heap \ls{h_1}.
\begin{lstlisting}
let soundness (m:mst_repr 'a) (wp:state_wp 'a) 
  : Lemma (requires ((forall w. theta$_{\!\textsf{w}}$ m [= wp) /\ empty `witnessed_before` m)) 
          (ensures  (forall post h_0. wp post h_0 ==> (let (v,h_1) = run_mst wp m post h_0 in post v h_1))) = ()  
\end{lstlisting}

\newtext{The soundness proof is not completely axiom-free because it relies on a
library for the \ls{witnessed} tokens that is defined in terms of an axiomatic
presentation of hybrid modal logic \cite{reed09} in \fstar.}

\section{Labeled references}
\label{sec:labeled-references}

We now explain how we build our reference labeling mechanism on top of Monotonic
State.

\subsection{The labels}

As illustrated in \autoref{sec:keyideas},  
in \secrefstar we work with three labels for references:
\begin{lstlisting}
type label = | Private | Shareable | Encapsulated
\end{lstlisting}
When a reference is allocated, it starts labeled with \ls{private}.
The user can then relabel the reference either with \ls{shareable} or \ls{encapsulated}.
Labeling a reference with \ls{shareable} or \ls{encapsulated} is permanent. We ensure this
using the following preorder on reference labels:
\begin{lstlisting}
let $\hrel_{\text{\sffamily lab}}$ : preorder label = fun l l' -> match l , l' with
  | Private , _ | Shareable , Shareable | Encapsulated , Encapsulated -> True | _ , _ -> False
\end{lstlisting}

To track references' labels and how they evolve in programs, we assume a ghost top-level monotonic reference\footnote{
  \newtext{We prefer to assume the top-level reference to avoid precisely modeling ghost state~\mbox{\cite{HOghoststate,pulsecore}}.}}
that stores a map from the addresses of references to their labels
\begin{lstlisting}
type label_map_t = erased (mref (pos -> erased label) (fun m_0 m_1 -> forall r. (m_0 r) $\hrel_{\text{\sffamily lab}}$  (m_1 r)))
assume val label_map : label_map_t
\end{lstlisting}
The heap relation $\preceq$ defined in \autoref{fig:monotonic-state}
from \autoref{sec:monotonic-state} ensures
that \ls{label_map} can only evolve according to $\hrel_{\text{\sffamily lab}}$, which in turn is what makes being \ls{shareable} or \ls{encapsulated} a stable property of references: 
\begin{lstlisting}  
let is_shareable : mref_stable_heap_predicate = fun #a #rel (r:mref a rel) h ->
  h `contains` label_map (* necessary to show stability *) /\ Shareable? ((sel h label_map) (addr_of r))
let is_encaps : mref_stable_heap_predicate = fun #a #rel (r:mref a rel) h ->
  h `contains` label_map (* necessary to show stability *)  /\ Encapsulated? ((sel h label_map) (addr_of r))
\end{lstlisting}
where \ls{mref_stable_heap_predicate} is an \ls{mref}-parameterized stable heap
predicate. In contrast, as \ls{private} references can get 
relabeled, \newremove{with other labels,}%
being \ls{private} is {\em not} a stable property of the heap:
\begin{lstlisting}
let is_private_addr : pos -> heap -> Type0 = fun r h -> Private? ((sel h label_map) r)
let is_private : mref_heap_predicate = fun #a #rel (r:mref a rel) h -> is_private_addr (addr_of r) h
\end{lstlisting}

We mark the contents of the map and the map itself as \ls{erased}, to on
the one hand ensure that labels can only be used in specifications and programs
cannot branch on them, and on the other hand, to avoid passing \ls{label_map} to
computationally relevant operations that expect non-erased arguments. The latter
allows \ls{label_map} to be erased in extracted code (as \ls{erased} is
extracted to \ls{unit}).

Finally, we find it useful to define variants of the \ls{modifies}
clause discussed in \autoref{sec:monotonic-state} to specify that the footprint
of a computation modifies only \ls{shareable} or \ls{encapsulated}, or both kinds
of references:
\begin{lstlisting}
let modif_only_shareable_and_encaps (h_0:heap) (h_1:heap) : Type0 = forall a rel (r:mref a rel).
  (~(eq_addrs r label_map) /\ ~(is_shareable r h_0 \/ is_encaps r h_0)) ==> sel h_0 r == sel h_1 r
let modif_shareable_and (h_0:heap) (h_1:heap) (s:set nat) : Type0 = forall a rel (r:mref a rel).
  (~(eq_addrs r label_map) /\ ~(is_shareable r h_0 \/  (addr_of r) `mem` s)) ==> sel h_0 r == sel h_1 r
\end{lstlisting}

\subsection{The global invariant for shareable references}\label{sec:global-invariant}

We introduced our global invariant on the heap, that \ls{shareable} references can only point to \ls{shareable} references,
in~\autoref{sec:key-verification}.
To formally state it in \autoref{sec:lr-effect}, we want to quantify over all references of all types that are labeled as \ls{shareable} and
check that the stored value contains only \ls{shareable} references.
A convenient way to state the invariant is using a deep embedding of the types a \ls{shareable} reference can store values of.
\newremove{These are either base types, sums or products of them, other references storing values of such types, and data structures built from them,}
\newtext{These are built from base types, sums, products and references,}
collectively known in the literature as full ground types \cite{MurawskiT12,MurawskiT18,KammarLMS2017}.
\newtext{Arrows are not part of the embedding since our heap is only first-order (see~\autoref{sec:mst-representation}).}
The embedding is given by:
\begin{lstlisting}
type full_ground_typ = | SUnit | SInt | SBool
  | SSum : full_ground_typ -> full_ground_typ -> full_ground_typ
  | SPair : full_ground_typ -> full_ground_typ -> full_ground_typ
  | SRef : full_ground_typ -> full_ground_typ
  | SLList : full_ground_typ -> full_ground_typ
\end{lstlisting}
and comes with the evident recursive inclusion 
\ls{to_Type : full_ground_typ -> Type0} into \ls{Type0} types.

\newremove{It is important to emphasize that this deep embedding is only used for stating
the global invariant about \ls{shareable} references and to be able to do proofs about
its preservation, but it does not limit the capacity of the heap to store values
of more general types in universe \ls{Type0} (recall \mbox{\autoref{fig:monotonic-state}}).
Only certain
operations (\EG \ls{label_shareable} in
\autoref{sec:lr-effect}) need their type arguments to be given as
\ls{full_ground_typ} instead of \ls{Type0}, to help \fstar with proving that the
global invariant is preserved.}

To check that a stored value contains only \ls{shareable} references,
we use the deep embedding to define a  
predicate \ls{forall_refs} which checks that every reference in a given value of \ls{full_ground_typ} type 
satisfies a predicate. We then lift this predicate to properties of references with respect to heaps.
As these predicates are parametric in \ls{pred}, we can instantiate them
with different predicates. 
\begin{lstlisting}
let rec forall_refs (pred:mref_predicate) (#t:full_ground_typ) (x:to_Type t) : Type0 =
  let rcall #t x = forall_refs pred #t x in
  match t with | SUnit | SInt | SBool -> True
  | SSum t_1 t_2 -> match x with | Inl x' -> rcall x' | Inr x' -> rcall x'
  | SPair t_1 t_2 -> rcall x._1 /\ rcall x._2
  | SRef t' -> pred x
  | SLList t' -> match x with | LLNil -> True | LLCons v xsref -> rcall v /\ pred xsref
let forall_refs_heap (pred:mref_heap_predicate) h x = forall_refs (fun r -> pred r h) x
\end{lstlisting}
\ca{\ls{forall_refs} is not used anywhere. only \ls{forall_refs_heap} is used. lets see if in 5 we're going to use it}

Finally, to be able to define the invariant, we transitively close these predicates on a heap:
\begin{lstlisting}
let trans_forall_refs_heap (h:heap) (pred:mref_stable_heap_predicate) =
  forall (t:full_ground_typ) (r:ref (to_Type t)). h `contains` r /\ pred r h ==> forall_refs_heap pred h (sel h r)
\end{lstlisting}
This general version we apply to
\ls{is_shareable} and \ls{contains} predicates, \EG so as to specify that every \ls{shareable}
reference \ls{r} points to a value \ls(sel h r) that contains only \ls{shareable} references.

\subsection{The \ls{LR} effect}\label{sec:lr-effect}

We define the Labeled References effect \ls{LR} on top of the Monotonic State
effect \ls{MST} by restricting \ls{MST} computations with a global invariant that
enforces the correct usage of shareable references and the ghost state storing
reference labels. The invariant is defined as a conjunction of conditions: 
\da{Do we want to use some nice calligraphic notation for this invariant?}
\begin{lstlisting}
let contains_pred #a #rel (r:mref a rel) h = h `contains` r
let lr_inv (h:heap) : Type0 =
  trans_forall_refs_heap h contains_pred /\ (* contained references point to contained references *)
  trans_forall_refs_heap h is_shareable /\ (* shareable references point to shareable references *)
  h `contains` label_map /\ (* the labeling map is contained in the heap *)
  is_private (label_map) h /\ (* the labeling map remains private *)
  (forall r. r >= next_addr h ==> is_private_addr r h) (* all yet to be allocated references are labeled as private *)
\end{lstlisting}
The last condition ensures that any freshly allocated reference is \ls{private}
by default and will only become \ls{shareable} by explicitly labeling it so,
allowing for fine-grained control over which references can cross the
verified-unverified code boundary. This also applies to \ls{encapsulated} references.

The \ls{LR} effect is then defined on top of \ls{MST} by packaging the invariant
both into pre- and post-conditions. By being an effect synonym as opposed to a
new effect, we can seamlessly coerce any \ls{MST} computation into \ls{LR}, as
long as it satisfies the invariant, something that is very helpful for linking
verified code with unverified code (see \autoref{sec:key-ucode}).
In detail, \ls{LR} is defined in \fstar as follows:
\begin{lstlisting}
effect LR (a:Type) pre post = MST a (requires (fun h_0 -> lr_inv h_0 /\ pre h_0)) 
                                 (ensures (fun h_0 r h_1 -> lr_inv h_1 /\ post h_0 r h_1))
\end{lstlisting}

By providing the \ls{MST} operations we discussed in
\autoref{sec:monotonic-state} with suitable stronger pre- and post-conditions, we
can lift them to the \ls{LR} effect. For allocation, the \ls{LR}
operation \ls{lr_alloc #a #rel init} additionally requires that if the type
\ls{a} happens to be a \ls{full_ground_typ}, then every reference contained in the
initial value \ls{init} has to be contained in the heap. In its post-condition,
\ls{lr_alloc} additionally guarantees that the freshly allocated reference is
labeled \ls{private} and no other references become \ls{shareable} during its allocation. Modulo
the use of lemmas proving that allocation of references preserves \ls{lr_inv}, 
\ls{lr_alloc} is defined simply as
\ls{alloc}, with the type
\begin{lstlisting}
val lr_alloc #a (#rel:preorder a) (init:a) : LR (mref a rel)
    (requires (fun h_0 -> forall t. to_Type t == a ==> forall_refs_heap contains_pred h_0 #t init))
    (ensures (fun h_0 r h_1 -> alloc_post #a #rel init h_0 r h_1 /\ is_private r h_1 /\ becomes_shareable !{} h_0 h_1))
\end{lstlisting}

The specification of \ls{LR}'s \ls{lr_read} operation remains the same
as for \ls{MST}'s \ls{read} (see \autoref{fig:monotonic-state}). For
\ls{lr_write #a #rel r v}, we strengthen the pre-condition by additionally
requiring that the reference is different from the ghost state holding the
labeling map, and that if \ls{a} happens to be a \ls{full_ground_typ}, then every
reference contained in \ls{v} has to be contained in the heap, and that if \ls{r}
happens to be a \ls{shareable} reference, then every reference contained in
\ls{v} has to be \ls{shareable} as well. In return, the post-condition additionally
guarantees that no new references become \ls{shareable} while writing into \ls{r}. As
with \ls{lr_alloc}, modulo the use of some lemmas, \ls{lr_write} is defined as \ls{write}.
\begin{lstlisting}
val lr_write #a (#rel:preorder a) (r:mref a rel) (v:a) : LR unit
    (requires (fun h_0 -> write_pre #a #rel r v /\ ~(eq_addrs r label_map) /\ 
      (forall t. to_Type t == a ==> forall_refs_heap contains_pred h_0 #t v /\
        (is_shareable r h_0 ==> forall_refs_heap is_shareable h_0 #t v))))
    (ensures (fun h_0 () h_1 -> write_post r v h_0 () h_1 /\ becomes_shareable !{} h_0 h_1))  
\end{lstlisting}
Similarly to reading from a reference, witnessing and recalling stable
predicates also retain the exact pre- and post-conditions as listed in
\autoref{fig:monotonic-state}, as they also keep the whole heap unchanged.

In programs and in particular in the examples in the rest of the paper,
programmers can use \ls{lr_alloc}, \ls{lr_read}, and \ls{lr_write} via the
usual ML-style syntax, as \ls{alloc}, \ls{!}, and \ls{:=} respectively.

In addition to the operations lifted from the \ls{MST} effect, the \ls{LR} effect comes with two important additional operations, 
which allow \ls{private} references to be relabeled as \ls{shareable} or \ls{encapsulated}.
\begin{lstlisting}
val label_shareable (#t:full_ground_typ) (r:ref (to_Type t)) : LR unit
    (requires (fun h_0 -> h_0 `contains` r /\ is_private r h_0 /\ forall_refs_heap is_shareable h_0 (sel h_0 r)))
    (ensures (fun h_0 () h_1 -> equal_dom h_0 h_1 /\ modifies !{label_map} h_0 h_1 /\ 
      is_private (label_map) h_1 /\ becomes_shareable !{sr} h_0 h_1))
val label_encapsulated (#a:Type) (#rel:preorder a) (r:mref a rel) : LR unit
    (requires (fun h_0 -> h_0 `contains` r /\ is_private r h_0 /\ ~(eq_addrs r label_map)))
    (ensures (fun h_0 () h_1 -> equal_dom h_0 h_1 /\ modifies !{label_map} h_0 h_1 /\ 
      is_private (label_map) h_1 /\ becomes_encapsulated !{r} h_0 h_1))
\end{lstlisting}
Both operations can be called with private references contained in the
heap, with \ls{label_shareable} requiring that the reference \ls{r} has to store
\ls{full_ground_typ}-values (this helps to prove that \ls{label_shareable}
preserves the global invariant, as noted in \autoref{sec:global-invariant}),
that \ls{r} has a trivial preorder (via the type \ls{ref} of non-monotonic
references, see \autoref{fig:monotonic-state}), and that any references \ls{r}
points to have already been labeled as \ls{shareable}. This last point
means that when labeling complex data structures as \ls{shareable} (\EG linked
lists), we have to first apply \ls{label_shareable} at the end of chains of
references. The \ls{label_encapsulated} operation on the other hand applies to
\ls{private} monotonic references with arbitrary preorders, which in turn
necessitates the explicit comparison with \ls{label_map}---something that
follows implicitly for \ls{label_shareable} due to the preorders of \ls{r} and
\ls{label_map} being different. 
These operations are defined by appropriately modifying the ghost state storing
the labeling. \da{This part needs a better explanation, that the label map
modifications happen behind the abstract interface and on the unerased view of
the ghost state.}

\newtext{
\paragraph{Usability.}
From our experience working on the examples in the paper, we report that
the library interface is usable and, compared to the Monotonic State interface,
the extra burden to the user is reasonable for the extra functionality.
The main two differences are that the user has to 
specify both the footprint and what gets shared for each call,
and to keep track of the labels of the references.
Some of the details of the library are explicit into the interface, in particular predicates
about the \ls{label_map} reference, which are not relevant to the user,
but the user does not have to deal with them because of \fstar{}'s SMT automation.
We did not find a way to hide these details without losing expressiveness.
On the other hand,
the use of a deep embedding to implement the full ground types---which was 
necessary to prove preservation of the invariant inside the stateful operations---is not great for SMT automation,
as sometimes the user has to help \fstar{} figure out what the type is.
}

\section{Higher-order contracts}\label{sec:hoc}

As explained in \autoref{sec:key-hoc}, \secrefstar{} uses higher-order contracts at the 
boundary between verified and unverified code to ensure properties about 
returned or passed shareable references that cannot be statically determined.
The higher-order contracts
are built around two functions, \ls{import} and 
\ls{export}, that map from \emph{intermediate} types to \emph{source} types, and 
vice versa (\autoref{fig:overview}).
We first introduce intermediate types in \autoref{sec:hoc-interfaces}, and then
explain in \autoref{sec:import-export} how the \ls{export} and \ls{import} functions are defined.

\subsection{The polymorphic and intermediate interfaces}
\label{sec:hoc-interfaces}
The verified program and the unverified code communicate through an interface, 
which is different in the source language compared to the target language.
To represent the interface of unverified code as presented in \autoref{sec:key-ucode},
we define in \autoref{fig:concrete-predicates} the type class \ls{poly_iface},
parameterized by three abstract predicates (\ls{threep}).\tw{TODO: Explain also the witnessable bit, not introduced yet.}
The types that can appear on the interface of unverified code are thus
determined by the instances we provide for the \ls{poly_iface} class.
\newtext{We provide instances for base types, sums, products, references, and {\em arrow types}.
The instance for arrows is defined so that
\ls{poly_iface} can be used
to describe higher-order interfaces (\EG \ls{poly_student_hw} from \autoref{sec:key-ucode}).}
Notice that the instances for references and linked lists take
a full ground type (\autoref{sec:global-invariant}).
All types that have a representation as a full ground type can be
characterized by the \ls{poly_iface} class.

\begin{figure}
\begin{lstlisting}[escapechar=\$]
type threep = $\inv$:(heap -> Type0) & $\prref$:(#a:_ -> #rel:_ -> mref a rel) & $\hrel$:(preorder heap)
class poly_iface (preds:threep) (t:Type) = { wt : every_mref_tc t }
instance poly_iface_int preds : poly_iface preds int
instance poly_iface_sum preds a b {|poly_iface preds a|} {|poly_iface preds b|}:poly_iface preds (either a b)
instance poly_iface_ref preds (t_emb:full_ground_typ) : poly_iface preds (ref (to_Type t_emb)) 
instance poly_iface_arrow $\inv$ $\prref$ ($\hrel$) a b {| c1:poly_iface ($\inv$, $\prref$, $\hrel$) a |} {| c2:poly_iface ($\inv$, $\prref$, $\hrel$) b |}
  : poly_iface preds (x:a -> MST b (requires (fun h_0 -> $\inv$ h0 /\ c1.wt.$\every$ x $\prref$))
                                 (ensures (fun h_0 r h_1 -> $\inv$ h1 /\ h0 $\hrel$ h1 /\ c2.wt.$\every$ $\prref$ r)))
\end{lstlisting}
  \caption{Showing the polymorphic interface implemented with type class \ls{poly_iface}. Selected instances.}
  \label{fig:concrete-predicates}
\end{figure}

Having defined polymorphic interfaces necessary to represent unverified code,
we define intermediate interfaces as
polymorphic interfaces instantiated with some concrete predicates, 
and intermediate types as instances of such interfaces.
Because we represent unverified code using polymorphism (\autoref{sec:key-ucode}), our higher-order contracts have to be
enforced on functions under polymorphism as well \cite{sciostar}, which means that our contracts
do not know what they enforce until the very end, when the type is instantiated
and the actual checks are provided.

\subsection{Exporting and importing}
\label{sec:import-export}

The functions \ls{export} and \ls{import} are implemented using two type
classes, named \ls{exportable_from} and \ls{importable_to}.
\newtext{To build an instance for \ls{exportable_from}, one has to provide a function
from a strong type (\ls{styp}) to an intermediate type (\ls{ityp}), and a proof that it preserves
the $\every$ combinator.
Such a proof is easy to provide because a reference is not affected by the \ls{export}
function since its type is restricted to be a full ground type---\IE exporting a reference is identity.}
\newremove{For instance, \ls{exportable_from styp} carries an intermediate type \ls{ityp}
together with a function \ls{styp -> ityp} that is shown to preserve the $\every$
combinator.}
Its dual, \ls{importable_to} is similar, but it has to
account for potential breaches of contract by allowing \ls{import} to fail 
with an error.
We also provide a \ls{safe_importable_to} class for types for which importing 
always succeeds.

To be able to export/import functions that are polymorphic in the three predicates (\ls{threep}),
the type classes are parameterized with some implementation of the three predicates.
The \ls{specs_of_styp} is just an inductive type that semi-syntactically represents
the pre- and post-condition that appear on type \ls{styp}, structured as a tree.
It allows us to be able to
provide the contracts at importing/exporting time, when the three predicates become concrete.
The functions \ls{export} and \ls{import} take the contracts as a tree of type \ls{hoc_tree specs_of_styp},
which has the same structure as \ls{specs_of_typ}, and its type guarantees that each pre- and post-condition that has to
be enforced has a contract.

\begin{lstlisting}
class exportable_from (preds:threep) (styp:Type) (specs_of_styp:spec_tree preds) = {
  c_styp : every_mref_tc styp; 
  ityp : Type; c_ityp : poly_iface preds ityp;
  export : styp -> hoc_tree specs_of_styp -> ityp;
  export_preserves_$\prref$ : x:styp -> $\prref$:_ -> hocs:_ -> Lemma ($\every$ $\prref$ x) ($\every$ $\prref$ (export x hocs)) }
class importable_to (preds:threep) (styp:Type) (specs_of_styp:spec_tree preds) = {
  c_styp : every_mref_tc styp;
  ityp : Type; c_ityp : poly_iface preds ityp;
  import : ityp -> hoc_tree specs_of_styp -> either styp err;
  import_preserves_$\prref$ : x:ityp -> $\prref$:_ -> hocs:_ -> Lemma ($\every$ $\prref$ x) ($\every$ $\prref$ (import x hocs)) }
\end{lstlisting}

For intermediate types (and thus already in the source language),
these type classes are trivially instantiated with identities, and for
operators such as \ls{option}, one proceeds recursively. For instance, the 
\ls{export} function for \ls{option a} when {a} is itself exportable will 
\emph{map} the \ls{export} function for \ls{a}.
What is more interesting is importing and exporting functions that may have
pre- and post-conditions.

\paragraph{Exporting functions.} 
We provide a general instance to export a function \ls{f} that has source type \ls{a -> MST (either b err) pre post}
to a function of the intermediate type \ls{a' -> MST (either b' err) pre' post'},
where \ls{a'} and \ls{b'} are the intermediate counterparts of \ls{a} and \ls{b},
and \ls{pre'} and \ls{post'} are the pre- and post-condition that appear in
\ls{poly_iface_arrow} in \autoref{fig:concrete-predicates}.
For this, we need an \ls{import} function (\ls{a' -> a}) for type \ls{a} and an
\ls{export} function (\ls{b -> b'}) for type \ls{b}, which we then compose with \ls{f}.
To be able to call \ls{f}, we have to ensure that \ls{pre} holds,
which we cannot do in general without a dynamic check, thus, the instance
requires the existence of a
\ls{check} function which approximates the pre-condition:
\begin{lstlisting}
check : x:a -> MST (either unit err) (requires (pre' x)) (ensures (fun h_0 r h_1 -> h_0 == h_1 /\ (Inl? r ==> pre x h_0)))
\end{lstlisting}
which either returns an error or succeeds by enforcing \ls{pre}. 
Notice that \ls{check} may safely assume \ls{pre'} as it  
holds when it is called.
Finally, one has to ensure that \ls{post'} holds, for which
we require that the original post-condition of \ls{f} satisfies the following property:
\ls{forall x h_0 r h_1. post x h_0 r h_1 ==> post' h_0 r h_1}.
One detail that we omitted is that \ls{pre'} and \ls{post'} are
polymorphic in the type of \ls{x} and \ls{r},
compared to \ls{pre} and \ls{post} that use the types \ls{a} and \ls{b}.
Therefore, because we call \ls{import}
and \ls{export}, we also have to prove that
\ls{pre' x ==> pre' (import x)} and \ls{post' h_0 r h_1 ==> post' h_0 (export r) h_1},
which we do by using the two lemmas included in the type classes, \ls{export_preserves_}$\prref$
and \ls{import_preserves_}$\prref$.

\paragraph{Importing functions.}
Similarly, we can import a general function \ls{f} of the intermediate type \ls{x:a' -> MST (either b' err) pre' post'}
to the source type \ls{x:a -> MST (either b err) pre post}.
Because the type already accounts for errors, we use the 
\ls{safe_importable_from} class to avoid duplicating errors.
This time we export the input from \ls{a} to \ls{a'}, call the function, and import the output from 
\ls{b'} to \ls{b}. To be able to call the function \ls{f}, we need to enforce
\ls{pre'} using \ls{pre}, hence we only provide an instance under the assumption
\ls{forall x h_0. pre x h_0$\ $ ==> pre' x h_0}. Similarly, once the function is run,
we only know that \ls{post'} is verified, so we also require
\ls{forall x h_0 e h_1. pre x h_0 /\\ post' h_0$\ $ (Inr e) h_1$\ $ ==> post x h_0 (Inr e) h_1}, 
which ensures we do preserve the post-condition of the error case. Often this 
will be discharged automatically by having the post-condition trivially true for
errors.
The more interesting part is when the function succeeds, and we need to verify
the post-condition. For this, the instance also requires a 
stateful contract \cite{FindlerF02}, see \ls{select_check} \cite{sciostar} below:
\begin{lstlisting}
type cb_check a b {| every_mref_tc b |} pre post x eh =
  r:b -> MST (either unit err) (fun h_1 -> post' eh r h_1) (fun h_1 rck h_1' -> h_1 == h_1' /\ (Inl? rck ==> post x eh r h_1))
select_check : x:a -> MST (eh:erased heap {pre x eh} & cb_check a (either b err) pre post x eh) 
    (requires (pre x)) (ensures (fun h0 r h1 -> reveal r._1 == h_0 /\ h_0 == h_1))
\end{lstlisting}
A stateful contract is first run before calling \ls{f} to read
from the heap information necessary to enforce the post-condition, which it can store in the closure,
and then the callback of the contract is run after \ls{f} returns to verify
the post-condition \ls{post}, or return an error 
otherwise. The \ls{same_values} predicate from the autograder example (\autoref{fig:autograder_example})
is enforced by \secrefstar{} using a stateful contract.

\begin{figure}
  \begin{mdframed}[backgroundcolor=black!5,hidealllines=true]
  \begin{lstlisting}[xleftmargin=-6pt,framexleftmargin=0pt]
  let $\cinv$ = lr_inv (* from $\text{\autoref{sec:lr-effect}}$. The three concrete predicates, defined as in $\autoref{sec:key-ucode}$.*)
  let $\cprref$ = fun r -> witnessed (is_contained r) /\ witnessed (is_shareable r)
  let $\chrel$ = fun h_0 h_1 -> modif_shareable_and_encaps h_0 h_1 /\ same_labels h_0 h_1
  
  type interface^S = {
    ctype : threep -> Type;
    specs_of_ctype : threep -> spec_tree threep;
    hocs : hoc_tree (specs_of_ctype ($\inv_c$, $\prref_c$, $\hrel_c$));
    imp_ctype : preds:threep -> safe_importable_to preds (ctype preds) (specs_of_ctype preds);
    $\psi$ : heap -> int -> heap -> Type0 }

  type ctx$^S$ I$^S$ = I$\overset{S}{.}$ctype ($\inv_c$, $\prref_c$, $\hrel_c$)
  type prog$^S$ I$^S$ = I$\overset{S}{.}$ctype ($\inv_c$, $\prref_c$, $\hrel_c$) -> LR int $\top$ I$\overset{S}{.}\psi$
  type whole$^S$ = $\psi$ : post_cond & (unit -> LR int $\top$ $\psi$)
  let link$^S$ #I$^S$ (P:prog$^S$ I$^S$) (C:ctx$^S$ I$^S$) = (| I$\overset{S}{.}\psi$, fun () -> P C |) (** denoted by $C[P]$ **)

  type interface^T  =  {
    ctype : threep -> Type;
    interm_ctype : preds:threep -> poly_iface preds (ctype preds); }
    
  type ctx$^T$ I$^T$ = $\inv$:_ -> $\prref$:_ -> $\hrel$:_ -> 
    poly_alloc_t $\threep$ -> poly_read_t $\threep$ -> poly_write_t $\threep$ -> I$\overset{T}{.}$ctype ($\inv$, $\prref$, $\hrel$)

  type prog$^T$ I$^T$ = I$\overset{T}{.}$ctype ($\inv_c$, $\prref_c$, $\hrel_c$) -> LR int $\top$ $\top$
  type whole$^T$ = unit -> LR int $\top$ $\top$
  let link$^T$ #I$^T$ (P:prog$^T$ I$^T$) (C:ctx$^T$ I$^T$) =
    fun () -> P (C $\cthreep$ ctx_alloc ctx_read ctx_write) (** denoted by $C[P]$ **)  

  let compile_interface (I$^S$:interface$^S$) : interface$^T$ = { (** denoted by $\cmp{I^S}$ **)
    ctype = (fun preds -> (I$\overset{S}{.}$imp_ctype preds).ityp);
    interm_ctype = (fun preds -> (I$\overset{S}{.}$imp_ctype preds).c_ityp); }
  let compile_prog #I$^S$ (P:prog$^S$ I$^S$) : prog$^T$ (I$^S\cmparrow{}$) = (** denoted by $\cmp{P}$ **)
    fun C -> P ((I$^S$.imp_ctype ($\inv_c$, $\prref_c$, $\hrel_c$)).safe_import C I$^S$.hocs)
  let back_translate_ctx #I$^S$ (C:ctx$^T$ I$^S\cmparrow{}$) : ctx$^S$ I$^S$ = (** denoted by $C^T\bakarrow{}$ **)

    (I$^S$.imp_ctype ($\inv_c$, $\prref_c$, $\hrel_c$)).safe_import (C ($\inv_c$, $\prref_c$, $\hrel_c$)) I$^S$.hocs
  \end{lstlisting}
  \end{mdframed}
  \caption{Secure compilation framework. Idealized mathematical notation.
    Type \ls{threep} is presented in \autoref{fig:concrete-predicates}.
    Types \ls{poly_alloc_t}, \ls{poly_read_t}, and \ls{poly_write_t} are presented in \autoref{fig:type-poperations}.
  \ca{too similar to popl version}}
  \label{fig:model_sec_comp}
\end{figure}

\section{\titsecrefstar{}: Formally Secure Compilation Framework}
\label{sec:secure-compilation}

In this section we present the formalization of \secrefstar{} (\autoref{fig:model_sec_comp}).
The formalization is mechanized in \fstar{} 
and instantiates the compiler model presented by \citet{AbateBGHPT19}
with source and target languages that are shallowly embedded in the style
of \citet{sciostar}.
We start by presenting the source (\autoref{sec:sc-source}) and target language (\autoref{sec:sc-target}).
We continue by giving a semantics to them (\autoref{sec:sc-semantic}).
Then, we prove soundness and RrHP by carefully defining compilation and back-translation (\autoref{sec:sc-soudness-rrhp}).
For all of these, we first present the setting when the verified program has initial control, however, \secrefstar{} is verified to be
secure also when the unverified code has initial control (also \autoref{sec:sc-soudness-rrhp}).

\subsection{Source Language}\label{sec:sc-source}

We make a clear split between the verified program and
the code it is linked with.
In the source language, we refer to the former as a partial source program, and the
latter as a source context. The latter is also a verified piece of code,
but for secure compilation we think of it as a target context (aka unverified) that was back-translated (the dual of compilation)
to a source context (aka verified).

The partial source program and context communicate through a higher-order interface
of type \ls{interface}$^S$, which contains a type (\ls{ctype}), a type class constraint on this type (\ls{imp_ctype}), and a post-condition $\psi$.
As we want the type \ls{ctype} to be polymorphic in the three predicates in our source language, so as 
to obtain a polymorphic type in the target language, we include
the fields \ls{specs_of_ctype} and \ls{hocs} to be able to apply higher-order contracts
to a polymorphic type as explained in \autoref{sec:import-export}.
The representation of the partial source programs (\ls{prog}$^S$) is a computation in the monadic effect \ls{LR} that
takes as argument a value of type \ls{ctype} instantiated with the concrete predicates,
returns an integer, and ensures the post-condition $\psi$.
The representation of source contexts (\ls{ctx}$^S$) is just a value of type \ls{ctype}
instantiated with the concrete predicates.\ca{what if it is a ref?}
Source linking (\ls{link}$^S$) is then defined just as function application.
A whole source program is a dependent pair between a post-condition and a computation in effect \ls{LR}
from unit to integer that guarantees the post-condition.

The type class constraint \ls{imp_ctype} guarantees that the specification on the interface
can be soundly enforced by \secrefstar{}. The assumptions made about the unverified code
appear on this source interface (\EG the assumptions the autograder makes about
the homework in \autoref{fig:autograder_example}),
and \ls{imp_ctype} ensures that they are one of the two types of
assumptions that \secrefstar{} supports (\autoref{sec:key-verification}).
\ca{maybe more details are needed on why ``importable'' does that}

\subsection{Target language}\label{sec:sc-target}

The formalization of the target language is similar to the source language,
with two notable differences:
the representation of partial target programs (\ls{prog}$^T$) does not require the 
computation to be verified with respect to a post-condition, since it has the
trivial post-condition;
and 
target contexts (\ls{ctx}$^T$) use the
representation 
from \autoref{sec:key-ucode}.
The type \ls{ctype} in the target interface (of type \ls{interface}$^T$),
now is constrained with the \ls{poly_iface} type class from \autoref{sec:hoc-interfaces}.


%
A partial target program (see \ls{prog}$^T$) expects the context to satisfy
the concrete specification it assumes about the context: to only modify shareable and encapsulated references.
Therefore, linking is done at the types of the target language, but after linking,
the communication between the program and 
the context happens at an intermediate type that contains specifications~\cite{sciostar}.
This can be seen in the definition of target linking (\ls{link}$^T$).
It instantiates the target context with the concrete predicates $\cinv$, $\cprref$, and
$\chrel$, thus strengthening the type of the context.
%

During linking, the concrete implementations of the operations \ls{alloc}, \ls{read}, and \ls{write} are provided.
As explained in \autoref{sec:key-ucode}, their definitions are straightforward (specifications and proofs elided):
\begin{lstlisting}
let ctx_alloc init = let r = lr_alloc init in label_shareable r; r / let ctx_read = lr_read / let ctx_write = lr_write
\end{lstlisting}

The representation of target contexts is a reasonable representation for unverified code---logical reasoning
is not necessary to prove the pre- and post-conditions that use the abstract predicates.
As noted in \autoref{sec:key-ucode}, we deeply embed a total
simply typed language with first-order ML-style references, and provide a translation
from any typed expression in this language to a context of type \ls{ctype}
to the representation of target contexts \ls{ctx}$^T$.
The translation is defined structurally recursively on the typing derivations.
Formalizing this translation revealed multiple edge cases in higher-order settings,
and also the fact that monotonic references cannot be shared with unverified code.
This gives us sufficient confidence that our representation is reasonable,
in addition to multiple examples we implemented using this representation in the artifact (\EG the homework in \autoref{sec:key-ucode}).

\ca{i think the point made in the previous paragraph, makes this paragraph rather not interesting.
disabling for now}\ca{
Another fact is that a target context that uses the abstract predicates is forced to use
the operations it gets as argument, because it is prevented to use the default operations
of the monadic effect \ls{LR}.
The prevention is done because the default operations of the monadic effect \ls{LR}
have concrete specifications that are not implied (or do not imply) the abstract predicates.
With the default operations, to allocate a reference or to write to a reference, one has to prove
that the modification done to the heap implies the abstract predicates $\inv$ and $\hrel$, which is not possible.
Not even reading from a reference is possible,
because to read, one has to prove that the reference is contained in the heap,
and since all the references the context has access to are refined with an abstract $\prref$,
we cannot do.
Therefore, the target context is forced to use the operations it gets as argument,
which guarantees that the unverified context works only with shareable references because
we know the operations we pass to it do that.}

\subsection{Semantics}\label{sec:sc-semantic}

Our source and target languages are shallowly embedded. Since whole source and
target programs are computations in the \ls{LR} effect, we can define the
same semantics for both.
We define this semantics as predicates over the initial heap, the final 
result (an integer), and the final heap. 
\begin{lstlisting}
type state_sem = heap -> state -> heap -> Type0
let ($\subseteq$) (s_1 s_2:state_sem) = forall h_0 r h_1. s_1 h_0 r h_1 ==> s_2 h_0 r h_1
\end{lstlisting}
%
To 
define the semantics of whole programs, we must 
first reveal the monadic representation of 
\ls{LR}
using the \ls{reify} construct provided by 
\fstar{} \cite{indexedeffects}.
%
We define the semantics function \ls{beh} by 
adapting the monad morphism $\theta$ from \autoref{sec:mst-representation} \cite{dm4all,sciostar}.
\begin{lstlisting}
let beh_mst (m:mst_repr int) : state_sem = fun h_0 r h_1 -> forall p. $\theta$ m p h_0 ==> p r h_1
let beh (W:unit -> LR int $\top$ $\top$) = beh_mst (reify (W ()))
\end{lstlisting}
Since we use $\theta$ for both the Dijkstra monad and our semantics function,
we can easily turn the intrinsic property that a whole source program
satisfies a post-condition $\psi$ into an extrinsic property:

\begin{theorem}[Soundness of whole source programs]
  \label{thm:sem-source-whole-programs}
  $\forall (\psi, W):\mathsf{whole}^S.\ \mathsf{beh}(W) \subseteq \psi$
\end{theorem}
\begin{proof}[Proof sketch.] For a whole source program,
  we know its underneath monadic representation is \ls{m:(mst_repr int)\{theta m} $\sqsubseteq$ \ls{(fun p h_0 -> forall r h_1.}$\psi\ $\ls{h_0 r h_1 ==> p r h_1)\}} (\autoref{sec:mst-representation}),
  which implies that \ls{beh m} $\subseteq\psi$. \end{proof}

The proof is so straightforward because the type of a whole source program 
contains an intrinsic specification. We just lift the specification to 
an extrinsic property of the whole program.
Intuitively, the proof that the program satisfies the specification
was already done modularly by \fstar{} typing.

\newtext{\paragraph{Trusted Computing Base (TCB).}}
The semantics function is written in the same language as the code it reasons about, 
meaning that the semantics a program is given 
depends on \fstar{}'s consistency.
Therefore, as part of our TCB we include the 
following \fstar{} features that \secrefstar{} relies on:
\fstar{}'s type system,
the effect mechanism (which has a paper proof), the axiomatized \ls{Ghost} effect, 
and \newtext{the type abstractions provided by the module system (\EG the implementation of references is hidden behind a module interface)}.
\newremove{We also assume that \ls{assume} is not used by the program and context.}
\ca{what happens if any of these is unsound? one can
  inhabit verified programs/contexts with unverified code.}

\subsection{Soundness and Robust Relational Hyperproperty Preservation (RrHP)}\label{sec:sc-soudness-rrhp}
\newtext{
\citet{sciostar}, in their setting for IO, show that soundness and RrHP
can be proved using a Syntactic Inversion Law,
informally, \textit{Compilation+Target Linking=Back-translation+Source Linking}.
Such a strong result is achievable because we are working with
shallow embeddings, 
we take advantage of polymorphism to model the context,
and we implement higher-order contracts using two dual functions.
We managed to adapt their compilation model in our setting for state,
and to prove the law and the two criteria by following their proofs.}
\ca{
The novel additions include dealing with references that cross the
verified-unverified boundary by taking advantage of the monotonicity of \ls{LR},
like $\prref$ in the representation for unverified code and how $\cprref$
uses \ls{witnessed}.}

We define the compilation of partial programs (\ls{compile_prog} in \autoref{fig:model_sec_comp}, denoted as $\cmp{P}$) as
wrapping the partial source program in a new function of type \ls{prog}$^T$.
The partial target program takes as argument an instantiated target context whose
specification states that it modifies only shareable references.
Therefore, to enforce the other assumptions made by the partial source program, \secrefstar{} adds to the compiled program the
higher-order contracts using \ls{safe_import}.

We define the back-translation of target contexts (denoted by $\bak{C}$) 
as instantiating the context and then adding the higher-order contracts to produce a source context.
One immediately observes that this is exactly what happens during compilation plus target linking,
\newtext{and since source linking is just function application, it means syntactic inversion law can be proved
just by unfolding definitions.}

\begin{theorem}[Syntactic inversion law]
  \label{thm:syntactic-equality}
  $\forall I^S.\ \forall C^T:\mathsf{ctx}^T\ \cmp{I^S}.\ \forall P:\mathsf{prog}^S\ I^S.\ C^T[\cmp{P}]	= \bak{C^T}[P]$
\end{theorem}

\secrefstar{} allows one to statically verify that a partial program satisfies
a post-condition $\psi$. We prove a soundness theorem that guarantees
that compiling the partial program and then linking it with arbitrary unverified code
produces a whole target program that also satisfies the post-condition $\psi$
(when the partial program has initial control).

\begin{theorem}[Soundness]
  \label{thm:soundness}
  $\forall I^S.\ \forall P:\mathsf{prog}^S\ I^S.\ \forall C^T:\mathsf{ctx}^T\ \cmp{I^S}.\ \mathsf{beh}(C^T[\cmp{P}])\subseteq I\overset{S}{.}\psi$
\end{theorem}
\begin{proof}[Proof sketch.]
  We rewrite the goal with \autoref{thm:syntactic-equality} to $C^T[\cmp{P}] = \bak{C^T}[P]$.
  Now, $\bak{C^T}[P]$ builds a whole source program and we can apply \autoref{thm:sem-source-whole-programs}.
\end{proof}

We prove that \secrefstar{} robustly preserves relational hyperproperties.
\citet{AbateBGHPT19} have this as their strongest secure compilation criterion, and
they show it is stronger than full abstraction.

\begin{theorem}[Robust Relational Hyperproperty Preservation (RrHP)]
  \label{thm:rrhp}
  $$\forall I^S.\ \forall C^T:\mathsf{ctx}^T\ \cmp{I^S}.\ \exists C^S:\mathsf{ctx}^S\ I^S.\ \forall P:\mathsf{prog}^S\ I^S.\ \mathsf{beh}(C^T[\cmp{P}])	= \mathsf{beh}(C^S[P])$$
  \end{theorem}
\begin{proof}[Proof sketch.]
  The proof follows by back-translating the target context ($\bak{C^T}$) and using it to instantiate the existential of a source context, after which we 
  conclude with \autoref{thm:syntactic-equality}.
\end{proof}

Finally, we have proved that \secrefstar{}
also satisfies RrHP (\autoref{thm:rrhp}) in two other settings as~\citet{sciostar}.
\textbf{The dual setting: when the context has initial control.}
In this setting, the proofs are similar,
since \autoref{thm:syntactic-equality} again holds.
Instead of soundness, we proved the dual of soundness,
a less interesting criterion stating that
the resulting whole target program modifies only shareable references.
\newtext{As a sanity check that our shallow embedding of unverified code
is correct,}
we also proved RrHP in \textbf{the setting where target contexts are
represented with typed syntactic expressions} of the 
deeply-embedded language mentioned in \autoref{sec:key-ucode},
a case in which target linking and back-translation first do the extra
step of translating the expression into the representation for target contexts (\ls{ctx}$^T$).
The proofs can be found in the artifact.

\section{Case study: simple stateful cooperative multi-threading}\label{sec:case-study}
To test our framework further, we verify a simple stateful
cooperative multi-threading case study.
In this case study, we are particularly interested in showcasing how scheduling
properties can be verified, even when the verified scheduler is orchestrating
tasks that are themselves unverified.

Our model of cooperative multi-threading is inspired by the work
of \citet{Dolan17} in the context of OCaml, where effect handlers are
used to program concurrent systems.
In this model, running threads are represented as computation trees.
A computation tree is either a final value
or a computation yielding control to the next task. When yielding control, the
continuation is a stateful computation that returns a new computation tree.
We can model these trees as follows:
\begin{lstlisting}
type atree (a:Type0) = | Return : a -> atree a | Yield : continuation a -> atree a
and continuation a = unit -> LR (atree a) (requires (fun h_0 -> True)) (ensures (fun h_0 _ h_1 -> h_0 $\chrel$ h_1)
\end{lstlisting}
\newremove{These trees share a global reference \ls{r},
which can be accessed by any of the involved stateful computations.}
Continuations represent unverified
code, \newremove{and thus they have a corresponding
specification:\ch{This sounds awkward: it's unverified thus it has a
spec?}  the pre-condition requires the reference \ls{r} to
be \ls{shareable}, while the post-condition ensures that}\newtext{and
their specification consists of a trivial precondition and a
postconditon stating} the unverified code only modifies \ls{shareable}
(or \ls{encapsulated}) references using the relation $\chrel$
from \newremove{\mbox{\autoref{fig:concrete-predicates}}}\newtext{\autoref{fig:model_sec_comp}}.
As before, we omit reference containment requirements from pre-
and post-conditions, along with fuel parameters used for proving
termination.

We model the scheduler state using a \ls{private} reference, which
stores (i) an execution history (a list of thread IDs), (ii) the ID of
the next thread to be executed, and (iii) a count of the number of
finished threads (to which we refer as \emph{inactive}). We refine the
scheduler state by adding a simple fairness property: each thread must
have taken either \ls{ticks} or \ls{ticks + 1} steps, ensuring fair
progress among all threads (with \ls{ticks} logically incremented on
each round of the round-robin progress):
\begin{lstlisting}
let fairness_ticks (limit : int) (ticks : nat) (hist : list int) (next : int) = ...
let fairness (limit : int) (hist : list int) (next : int) = exists  ticks : nat . fairness_ticks limit ticks hist next
\end{lstlisting}

\newtext{In the refined scheduler state, we also ensure that both the next
thread and inactive counts are within valid ranges.} The
argument \ls{k} is used to keep track of the number of threads \newremove{that
was}passed to the scheduler on initialization. In detail, we define it
as follows:
\begin{lstlisting}
type count_st k = hist:(list int) & next:int & inact:int{fairness k hist next /\ 0 <= next < k /\ 0 <= inact <= k}
let counter_preorder k : preorder (count_st k) = fun v v' -> v._1 `prefix_of` v'._1
let counter_t k = mref (count_st k) (counter_preorder k)
\end{lstlisting}

The scheduler is implemented as a structurally recursive function with the following type:
\begin{lstlisting}
val scheduler (r : ref int) (tasks:list (continuation unit)) (counter:counter_t (length tasks))
: LR unit (requires (fun h_0 -> is_private counter h_0 /\ is_shareable r h_0))
        (ensures (fun h_0 _ h_1 -> modif_shareable_encaps_and h_0 h_1 !{counter} /\ gets_shared empty h_0 h_1))
\end{lstlisting}
It takes a \newremove{list of tasks to be executed (represented
as \ls{unit}-typed continuations), a global shared
reference \ls{r},}\newtext{global shared
reference \ls{r} to an integer (which can be accessed by any of the involved
stateful computations), a list of tasks to be executed (represented
as \ls{unit}-typed continuations),} and the reference to the current
state of the scheduler (which is required to be \ls{private} in the
pre-condition). On each iteration, the scheduler picks the next task
to be run, runs it, and acts depending on whether the outcome was a
final value or a computation yielding control. We are guaranteed that
at each step we run a next task from the given task list and only
shared references are modified. This invariant allows us to verify
the \ls{fairness} property of the scheduler's state\newtext{, as we
know that the counter reference was not modified by the tasks}.  In
comparison to the post-condition of unverified code, notice that the
post-condition for the scheduler ensures that no labels are modified,
but we do not claim to have only modified shared and encapsulated
variables: the scheduler state gets updated as the scheduler runs the
tasks.

Building on the scheduler implementation, we can write a function to run a list of computations.
\begin{lstlisting}
let run (v:int) (tasks:list (ref int -> continuation unit){length tasks > 0}) : LR int
  (requires (fun h_0 -> forall_refs_heap contains_pred h0 v /\ forall_refs_heap is_shareable h0 v)) 
  (ensures (fun h_0 _ h_1 -> h_0 $\chrel$ h_1)) =
    let counter : counter_t (length tasks) = alloc (| [], 0, 0 |) in
    let r = alloc v in label_shareable r;
    let tasks = map (fun (f : ref int -> continuation unit) -> f r) tasks in
    scheduler r tasks counter;
    !r
\end{lstlisting}
Here we allocate two references: one that is shared by the different
tasks (and thus it is labeled \ls{shareable}) and one that is used to
store the scheduler's state (and which is kept \ls{private}).  We then
run the computations using the scheduler and in the end return the
final value of the shared reference. The scheduler starts with an
empty history, the first task selected, and zero inactive tasks.
\newtext{After its execution, the final history
satisfies the \ls{fairness} predicate, indicating that all tasks took
the same number of steps. However, this fact is not made
explicit in the postcondition of \ls{run}, which only focuses on the references
modified during the execution.}
%
%
As a final remark, we note that the function \ls{run} can also be called by
unverified code to cooperatively execute other unverified code. Importantly, \ls{run}
has the exact same post-condition as unverified code (as used in the
continuations).

\section{Related work}
\label{sec:related}


\newremove{We already mentioned in the introduction a list of techniques that allow
verification of a stateful program, however, they do not support sound verification
of a program that can then be linked with arbitrary unverified code with which it dynamically shares mutable references.}

\newtext{
\paragraph{Verification of stateful code.}
Labeled References are encoded in \fstar{} on top of Monotonic State,
and this enables specifying the footprint of a call using labels,
which allows us to tame the complexity of globally tracking the dynamically shared references.
As explained in \autoref{sec:key-verification}, the usual way of specifying the footprint of a call,
by providing the precise set of references that gets modified~\cite{preorders,Reynolds02,leino10dafny},
does not scale to higher-order unverified code when dynamic sharing of references is allowed.
This approach enabled verification of stateful \fstar{} programs linked with
unverified code. Similar approaches~\cite{SwaseyGD17,SammlerGDL20}
were used to build logics on Iris~\cite{iris} for verifying such programs,
with which we compare next.
}

\newtext{
\paragraph{Capabilities and robust safety.}
\citet{SwaseyGD17} introduce OCPL, a logic for verifying object capability patterns.
It can be used to verify concurrent, non-terminating, stateful programs
that are fully higher-order, including support for higher-order references.
OCPL has a notion of high and low references,
where high references can be wrapped in closures that are passed to the untrusted code,
while low references are considered shared between the trusted and untrusted code.
%
In OCPL, low references can only point to low references, and
callbacks that are crossing the verified-unverified boundary can take as argument
and return only low references.
This notion of high and low references seems to
correspond to the \ls{private/encapsulated} and \ls{shareable} labels in Labeled References.
%
%
\citet{SwaseyGD17} prove a robust safety theorem stating that code verified with
OCPL can be safely linked with arbitrary unverified code.
%
They then use OCPL to verify the robust safety of using wrappers like
dynamic sealing, the caretaker pattern and the membrane pattern.
%
%
By comparison, \secrefstar is a secure compilation framework using higher-order
contracts that dynamically enforce specifications about shared references.
Soundness of \secrefstar (\autoref{thm:soundness})
intuitively provides similar guarantees to OCPL's robust safety theorem.
\secrefstar additionally satisfies RrHP (\autoref{thm:rrhp}), the
preservation of a large class of properties including robust safety.
%

\citet{cerise} propose Cerise: a program logic for reasoning about
a low-level capability machine. \newremove{Cerise}\newtext{It} guarantees robust safety of verified code
that interacts with unverified code in multiple scenarios:
when the verified-unverified code share static allocated memory, 
when they use wrappers in the style of \citet{SwaseyGD17},
and when the unverified code can allocate dynamically its own memory.
To obtain safety in the last scenario, they use activation records, 
a technique involving saving the
content of the references that are not in the call's footprint,
and \newremove{then} restoring them after the call,
but \newremove{this}\newtext{it} is done dynamically and adds an extra cost\ca{---for \ls{prog} above,
the contents of \ls{secret} would be checked/restored twice, once for each call into the unverified code}.
In our work, unverified code can also allocate its own memory dynamically, and 
we show how to specify the footprint using labels so that there is
no need to perform \newremove{such}dynamic restoring of references
when calling into unverified code.
%
}

\newtext{
In a setting where, instead of ML-style references, one 
uses integer pointers that can be dereferenced, \citet{SammlerGDL20} showed that one can
enforce robust safety by using sandboxing. They allow sharing dynamically allocated pointers
between verified and unverified code under some conditions: 
they have to clearly distinguish between
the \textit{high heap} of the trusted code and the \textit{low heap} of the untrusted code,
something we do not have to do because our references are unforgeable.
}

\paragraph{Secure compilation.}
The approaches of \citet{AgtenJP15} and \citet{StrydonckPD21} support secure compilation of formally verified stateful programs
against unverified contexts.
\citet{AgtenJP15} securely compile modules verified with separation logic~\cite{Reynolds02},
but written in the C language, which normally lacks memory safety.
To overcome that, they add runtime checks at the boundary between verified and unverified code,\ch{in
  addition to just the dynamic contracts part;
  which for them is not higher-order? how much they can deal with call-backs, function pointers, etc?}
and also use two static notions of memory: local, which is never shared
with the unverified code, and heap memory, which the unverified code can read and write
arbitrarily because of the lack of memory safety.
The local memory is protected using hardware enclaves, and if one wants to
share data from local memory, one has to copy it to the heap memory.
To protect the heap memory when calling into the unverified code,
they first hash the complement of the assumed footprint of the call,
and when control is returned, they recompute the hash and compare it,
to make sure that the unverified code did not change heap locations outside the footprint.
What we do is different because we work with memory safe languages with ML-style
references, which act like capabilities.
We show that if the source and target language have capabilities,
one can avoid these expensive hash computations by getting the footprint right.
Moreover, we have only one notion of references, which can be kept private or be
dynamically shared with the unverified code by using the labeling mechanism.

\citet{StrydonckPD21} compile stateful code verified with separation logic to a target language with
hardware-enforced {\em linear} capabilities, which cannot be duplicated at runtime.
The main idea is to require the unverified context to return the linear capabilities back
so that the verified code can infer that the context cannot have stored them.
Linear capabilities are, however, an unconventional hardware feature that is
challenging to efficiently implement in practice~\cite{GeorgesGSTTHDB21}.
In contrast to their solution, \secrefstar{} provides soundness and security
in a target language with ML-style mutable references,
which can be implemented using standard capabilities~\cite{CHERI-ISA},
but this required us to overcome the challenge of unverified contexts that stash references.

\newtext{
\citet{DBLP:conf/icfp/NewBA16} propose a fully abstract compiler
from a simply typed \(\lambda\)-calculus with unit, sums, pairs, and recursive types
to a target language that has additional support for exceptions.
Their compiler is not, however, addressing ML-references or verified code as ours does.
The full abstraction property is a security criterion that was shown to be 
implied by the criterion we prove about our compiler, RrHP~(\autoref{thm:rrhp}),
by \citet{AbateBGHPT19}.
\citet{DBLP:conf/icfp/NewBA16} employ a comparable approach to ours to prove full abstraction,
involving 
a translation between two shallow embeddings.}

\paragraph{Gradual verification} is a combination of static and dynamic verification proposed by \citet{BaderAT18}.
For assertions that cannot be proven statically using the existing code annotations,
one can opt to let the tool insert dynamic checks that dynamically enforce the assertion. 
\citet{DiVincenzoC0} introduced Gradual C0 to do gradual verification for stateful programs.
Gradual C0 requires analyzing the unverified code, while our DSL does not.
Gradual C0 requires the unverified code to determine what checks to add, and it can be many,
while \secrefstar{} avoids adding checks to preserve logical invariants about private references.
\secrefstar{} adds higher-order contracts to enforce pre- and post-conditions at the boundary of
verified-unverified code.
In the absence of unverified code,
the dynamic checks are necessary to be able to establish logical predicates about shareable references.
\newremove{As future work, one could try to extend \secrefstar{} so that it statically analyzes unverified code
and discards the contracts using soft contract verification~\mbox{\cite{NguyenTH14,NguyenGTH18}}.}

\newtext{
\paragraph{Secure interoperability.}
There is also related work based on typed approaches to safe interoperability
between high and low-level code that share references. The main difference with
those works is that they look at safe interoperability between languages with
different semantics, while we look at code \textit{verified} in the traditional
Hoare/separation-logic style and unverified code.
\citet{DBLP:conf/pldi/PattersonPDA17} present a multi-language 
that supports interoperability between a simply typed functional language and
a typed assembly language. Their multi-language allows
integrating assembly code into the functional language, and reason about their
mix. \citet{DBLP:conf/snapl/PattersonA17} present the notion of linking types,
a more expressive way to annotate the interface, that enable one to
reason about behavioral equality
using only one language, even though the code may be linked with
code written in a different language with more expressive power.
\citet{DBLP:conf/icfp/0001W023} show how to soundly encapsulate foreign code
using linked types, so that the foreign code is unobservable\ch{what do you mean
  by ``unobservable''? sounds suspicious to me}
and the invariants of the language are undisturbed.}

\tbd{
  \begin{itemize}

  \item Amin Timany and Lars Birkedal. Reasoning About Monotonicity in Separation Logic. In CPP 2021, January 2021,
    \url{https://cs.au.dk/~timany/publications/pub_pages/2021_cpp_monotone_monoid/}
    CA: Not directly relatable. Maybe we find a place to cite it if it makes sense.
  
  \item Modular Denotational Semantics for Effects with Guarded Interaction Trees \url{https://dl.acm.org/doi/10.1145/3632854}%
  \ca{Their GITree doesn't seem to be a monad. Was this doubled checked?}

  \item \ch{and also wrt our SecurePtrs work (same assumptions but trusted code not verified);
  also need to relate to \cite{SammlerGDL20} and \url{https://dblp.org/pid/20/5607.html}}
\end{itemize}  
}

\meta{DONE:
\begin{itemize}
  \item Cerise\cite{cerise} \url{https://dl.acm.org/doi/10.1145/3623510}
    \begin{itemize}
      \item "We present a mathematical model and accompanying proof
      methods that can be used for formal verification of functional correctness of programs running on a capability
      machine, even when they invoke and are invoked by unknown (and possibly malicious) code"
      \item "The security properties we focus on are centered around
      memory compartmentalization, in particular, local state encapsulation."
      \item "a [static allocated] capability (respectively, a program) is “safe” if it cannot be
      used to invalidate an invariant. Hence, safe capabilities can be shared freely with unknown
      code. Safety of a capability is defined in the program logic as a unary logical relation"
      \item "value which is safe to share only gives transitive access to other values are safe to share, or
      code that is safe to execute (in the case of a sentry capability)." CA: same as us
      \item "A value which is safe to execute allows the machine to run while preserving logical invariants
      (by definition of ...., provided the registers contain safe values"
      \item "It is therefore safe as long as the [static allocated] words stored in the corresponding memory region are
        safe, and continue to be so when the memory gets modified"
      \item CA: They provide a logical predicate to verify programs that share static words with unverified code. The programs are
        written in something similar to Assembly. The result in section 4 is for statically allocated capabilities.
      \item Section 7 about "Dynamic Memory Allocation": 7.3 "We define a heap-based calling convention that uses malloc to dynamically allocate activation
      records. An activation record is encapsulated in a closure that reinstates its caller’s local state, and
      continues execution from its point of creation. Conceptually, our heap-based calling convention
      can be seen as a continuation-passing style calling convention (one passes control to the callee,
      giving it a continuation for returning to the caller). This is similar to the calling convention that
      was used for instance in the SML/NJ compiler to implement an extension of Standard ML with
      call/cc [Appel 1992] (in the setting of a traditional computer architecture).
      In the setting of a capability machine, our calling convention is furthermore secure in the sense
      that it enforces local state encapsulation. In other words, one can use it to pass control to unknown
      adversarial code, while protecting local data of the caller, thanks to the use of sentry capabilities
      to implement the continuation. "
    \end{itemize}

  \item Gradual verification. Mostly work by Jenna DiVincenzo 
  \url{https://dl.acm.org/doi/10.1145/3704808}, 
  \url{https://dl.acm.org/doi/10.1145/3428296}, 
  \url{https://drops.dagstuhl.de/entities/document/10.4230/LIPIcs.ECOOP.2021.3}
  \begin{itemize}
    \item No support for higher-order
    \item They don't seem to have a modifies clause
  \end{itemize}

  \item Formal Verification of Heap Space Bounds under Garbage Collection \url{https://cs.nyu.edu/~am15509/publications/thesis.pdf}\url{https://cs.nyu.edu/~am15509/publications/defense_slides.pdf}
    CA: I don't think it is relevant. There is no verification against unverified code.

\end{itemize}  
}

\section{Conclusions and Future Work}
\label{sec:future}


In this paper we introduce \secrefstar{}, a secure compilation framework that
allows for the sound verification of stateful programs
that can be compiled and then linked with arbitrary unverified code, importantly allowing
dynamic sharing of mutable references.
\secrefstar{} solves a problem that, to our knowledge, exists in all proof-oriented programming languages
(Rocq, Lean, Isabelle/HOL, Dafny, \fstar{}, etc).
We present our contributions in the context of \fstar{}, but we believe many of the ideas  could be
applied to other proof-oriented programming languages as well, especially to the dependently-typed ones.
Our monadic representation for Monotonic State, which uses a Dijkstra monad, can be ported
to other dependently-typed languages (\EG Dijkstra monads
also have been implemented in Rocq \cite{dm4all, SilverZ21}).
Once Monotonic State is ported, we believe one should be able to port the other contributions
we present in this paper too: Labeled References, the higher-order contracts,
the model for secure compilation, and the soundness and RrHP proofs.
Dijkstra monads work well in \fstar{} because of SMT automation and the effect mechanism of \fstar{},
but they are not the standard means to perform verification in other proof-oriented programming languages.
Dijkstra monads are however just a way to reason about monadic computations,
thus one could instead use other frameworks built to reason about monadic computations~\cite{Yoon2022,LetanR20,relational700,XiaZHHMPZ20}.

We believe that \secrefstar{} is a significant step towards the long-term goal
of building a secure compiler from \fstar{} to OCaml.
We see four orthogonal open problems that need to be solved to \newremove{get closer to}\newtext{reach} this goal:
\begin{inlist}
\item
{\em Modeling higher-order store} for effectful functions
around \fstar{}'s predicative, countable hierarchy of universes.
\citet{DKoronkevich2024} propose an acyclic higher-order store,
but this would not be sufficient to model realistic attackers written in OCaml
(\EG no support for Landin's Knot~\cite{LandinKnot}).
\newtext{
In \newremove{a different type theory than the one of \fstar{},}\newtext{a different type theory than that of \fstar,}
\citet{gitrees} adapt Interaction Trees to Guarded Type Theory to support cyclic higher-order stores.}

\item
{\em Modeling more of the effects of OCaml} by extending the representation of the
monadic effect, and then figuring out how this impacts the soundness and
security guaranteed by \secrefstar{}.
While this was already studied separately for the IO effect~\cite{sciostar},
how other effects would impact soundness and security is an open question.
\newtext{
For example,
adding non-termination would involve switching from the inductive definition of the Free monad
to a co-inductive definition (probably of Interaction Trees~\cite{XiaZHHMPZ20,SilverZ21}),
and probably would involve assuming that untrusted code can always loop.}

%


\item {\em End-to-end verified compilation to OCaml}.\ch{Almost everything below
  is about verified compilation, {\bf not} about secure compilation, so for now I
  renamed the item to ``verified''. Security is much more interesting than just
  verification though, so that's the harder open problem.}
While \secrefstar{} targets a language shallowly embedded in \fstar{},
at least two more steps are necessary to target OCaml,
both open research problems: verified quoting
\footnote{
   Verifying quoting in MetaRocq: \href{https://github.com/MetaRocq/metarocq/blob/a50b36cc7688df75f4932d7242a889c5df8e5bea/quotation/theories/README.md}{https://github.com/MetaRocq/metarocq/blob/coq-8.16/quotation/theories/README.md}}
(to go from the shallow embedding to a deep embedding),
and verified erasure specialized for a Dijkstra monad.
\newremove{For verified erasure,} \citet{coq-to-ocaml} present a verified correct erasure from Rocq to Malfunction (one of the intermediate languages of OCaml),
which we hope can be specialized to computations represented using monads and extended to verify security.
Their proof relies on the formalization of Rocq~\cite{SozeauBFTW20},
which means that before one is able to reproduce and extend their proof,
one would first have to give a formalization to \fstar{}.
  
\item {\em Encoding Labeled References on top of Separation Logic},
yet how to do Separation Logic in \fstar{} well
is itself a topic of ongoing research~\cite{steelcore,steel,pulsecore}. \da{How is (4) related to getting a secure compiler from \fstar to OCaml specifically?}\ch{Fair point}
\ca{The challenge of encoding Labeled References on top of SL is having a way to keep track of the
  labels of the references.}
\end{inlist}

\newtext{
Finally, it would be interesting to do a performance evaluation of the
added higher-order contracts.
To improve performance, one could try to extend \secrefstar{} so that it statically analyzes unverified code
and discards the contracts using soft contract verification~\cite{NguyenTH14,NguyenGTH18}.}

\meta{
  \textbf{Open challenges:}
  \begin{itemize}
    \item We're limited to first-order stores because universe constraints prevent us to store 
    Free/ITree monads in it. Check Zulip FStar dm4all --- Representation of div.

    \item In a concurrent setting, one would have to avoid data races. Would it be possible to have a label for that?
  \end{itemize}
}

\section*{Data Availability Statement}
This paper comes with an \fstar{} artifact that includes the full
implementation of \secrefstar{}, the machine-checked proofs, and
examples. The examples can be extracted to native OCaml code (without any monadic
representation).
The artifact is available on Zenodo~\cite{secrefstar_artifact_zenodo} and Github~\cite{secrefstar_artifact_github}.

\ifanon\else
\begin{acks}
%
This work was in part supported
by the \grantsponsor{1}{European Research Council}{https://erc.europa.eu/}
under \ifcamera\else ERC\fi{} Starting Grant SECOMP (\grantnum{1}{715753}),
by the German  Federal Ministry of Education and Research BMBF (grant 16KISK038, project 6GEM),
and by the Deutsche Forschungsgemeinschaft (DFG\ifcamera\else, German Research Foundation\fi)
as part of the Excellence Strategy of the German Federal and State Governments
-- EXC 2092 CASA - 390781972.
Exequiel Rivas was supported by the Estonian Research Council starting grant PSG749.
Danel Ahman was supported by the Estonian Research Council grant PRG2764.
\end{acks}
\fi

\appendix

\section{Additional Example: Guessing game}

To showcase the use of encapsulated references, in \autoref{fig:guess_example} we present another example of a verified program that plays the 
``guess the number I'm thinking of'' with an unverified player. The verified program gives to the
player a range and a callback that the player can use to guess the number.
The verified program saves the guesses of the player using an encapsulated monotonic reference
that has a preorder that ensures that no guess is forgotten (line \ref{line:guess-mref}).
\secrefstar{} does not add any dynamic checks using higher-order contracts during compilation
because the post-condition of the player follows from the universal property of unverified code (\autoref{sec:key-ucode}).

\begin{figure}[h]
  \begin{lstlisting}[numbers=left,escapechar=\$,xleftmargin=13pt,framexleftmargin=10pt,numberstyle=\color{dkgray}]
type cmp = | LT | GT | EQ
let post = fun h_0 _ h_1 -> modif_shareable_and_encaps h_0 h_1 /\ same_labels h_0 h_1
type player_t = (l:int * r:int * (guess:int -> LR cmp (fun _ -> True) post)) -> LR int (fun _ -> l < r) post
let play_guess (l pick r:int) (player:player_t) : LR (bool * list int) (fun _ -> l < pick /\ pick < r) post =
  let guesses : mref (list int) (fun l l' -> l `prefix_of` l') = alloc [] in label_encapsulate guesses; $\label{line:guess-mref}$
  let final_guess = 
    player (fun x -> guesses := guesses @ [x]; if pick = x then EQ else if pick < x then LT else GT) in
  guesses := guesses @ [final_guess];
  if pick = final_guess then (true, !guesses)
  else (false, !guesses)
  \end{lstlisting}
  \caption{The illustrative example of a function that plays "Guess the number I'm thinking of"}
  \label{fig:guess_example}
\end{figure}

\ifanon\clearpage\fi

\bibliographystyle{abbrvnaturl}
\bibliography{fstar}

\end{document}